\documentclass[a4paper]{article}

\RequirePackage[colorlinks
,citecolor=cyan,urlcolor=blue
]{hyperref}
\usepackage[utf8]{inputenc}
\usepackage[T1]{fontenc}

\usepackage{amsfonts}
\usepackage{mathrsfs}

\usepackage{amssymb}
\usepackage{amsthm}
\usepackage{amsmath}

\usepackage{dsfont}

\emergencystretch=\maxdimen
\binoppenalty=\maxdimen
\relpenalty=\maxdimen

\usepackage{enumitem}

\usepackage{mathtools}

\newcommand\N{
	\mathbb{N}
}

\newcommand\F{
	\mathcal{F}
}

\newcommand{\ev}[1]{
	\mathbb{E}\left[ #1 \right]
}

\newcommand{\evM}[2]{
	\mathbb{E}^\mathbb{#1}\left[ #2 \right]
}

\newcommand\R{
	\mathbb{R}
}

\newcommand{\qvar}[1]{
	\langle #1 \rangle
}
\newcommand{\var}{
	\text{var}
}

\newcommand{\ind}[1]{
	\mathds{1}_{#1}
}

\DeclarePairedDelimiter\abs{\lvert}{\rvert}%
\DeclarePairedDelimiter\norm{\lVert}{\rVert}%

\makeatletter
\let\oldabs\abs
\def\abs{\@ifstar{\oldabs}{\oldabs*}}
\let\oldnorm\norm
\def\norm{\@ifstar{\oldnorm}{\oldnorm*}}
\makeatother

\usepackage{amsthm}
\usepackage[capitalise]{cleveref}

\newtheoremstyle{break}
  {\topsep}{\topsep}%
  {\upshape}{}%
  {\bfseries}{}%
  {\newline}{}%
\theoremstyle{break}
\newtheorem{thm}{Theorem}[section]
\newtheorem{lem}[thm]{Lemma}
\newtheorem{prop}[thm]{Proposition}

\newtheorem{dfn}{Definition}[section]

\theoremstyle{remark}
\newtheorem{rmk}{Remark}

\theoremstyle{plain} %
\newcommand{\thistheoremname}{}
\newtheorem{genericthm}[thm]{\thistheoremname}

\newtheoremstyle{mystyle}%
  {}%
  {}%
  {\itshape}%
  {}%
  {\bfseries}%
  {.}%
  { }%
  {\thmname{#1}\thmnumber{ #2}\thmnote{ (#3)}}%

\crefalias{thm}{theorem}
\crefalias{lem}{lemma}
\crefalias{prop}{proposition}
\crefalias{cor}{corollary}
\crefalias{dfn}{definition}
\crefalias{conj}{conjecture}
\crefalias{exmp}{example}

\makeatletter
\def\cleartheorem#1{%
    \expandafter\let\csname#1\endcsname\relax
    \expandafter\let\csname c@#1\endcsname\relax
}
\makeatother

\cleartheorem{rmk}
\newtheorem{rmk}{Remark}[section]

\usepackage{graphicx}
\usepackage[automark]{scrlayer-scrpage}
\usepackage{titlesec}
\usepackage[toc]{appendix}
\usepackage[space]{grffile}
\usepackage{svg}

\clearpairofpagestyles
\ohead{\headmark}
\ofoot*{\pagemark}
\addtokomafont{pagenumber}{\footnotesize}

\let\oldpar\paragraph
\renewcommand{\paragraph}[1]{\oldpar{#1} \mbox{} \\}

\graphicspath{ {./graficos/} }

\providecommand{\keywords}[1]
{
  \small	
  \textbf{\textit{Keywords---}} #1
}

\newcommand\fOU{fractional Ornstein-Uhlenbeck }
\newcommand\OU{Ornstein-Uhlenbeck }
\newcommand\ML{Mittag-Leffler function}
\newcommand\st{(\textbf{s}, \textbf{t})}
\newcommand\sti{(\textbf{s}, \textbf{t}_i)}
\newcommand\Wp{W^\mathbb{P}}
\newcommand\Bp{B^\mathbb{P}}

\newtheorem{assumption}{Assumption}

\usepackage{algorithm}
\usepackage[noend]{algpseudocode}
\usepackage{placeins}
\usepackage{caption}

\makeatletter
\renewcommand{\Function}[2]{%
  \csname ALG@cmd@\ALG@L @Function\endcsname{#1}{#2}%
  \def\jayden@currentfunction{#1}%
}
\newcommand{\funclabel}[1]{%
  \@bsphack
  \protected@write\@auxout{}{%
    \string\newlabel{#1}{{\jayden@currentfunction}{\thepage}}%
  }%
  \@esphack
}
\makeatother
\numberwithin{equation}{section}

\begin{document}

\author{Henrique Guerreiro \footnote{Supported by FCT Grant SFRH/BD/147161/2019.}
\\
hguerreiro@iseg.ulisboa.pt
 \and João Guerra 
\footnote{Partially supported by the project CEMAPRE/REM-UiDB/05069/2020 - financed by FCT/MCTES through national funds.} 
 \\
jguerra@iseg.ulisboa.pt\\
}
\date{
ISEG - School of Economics and Management, Universidade de Lisboa \\
REM - Research in Economics and Mathematics, CEMAPRE \\
Rua do Quelhas 6, 1200-781 Lisboa, Portugal \\
\vspace{10pt}
\today 
}
\title{
VIX pricing in the rBergomi model under a regime switching change of measure
}

\bibliographystyle{apalike}

\maketitle
\begin{abstract}
The rBergomi model under the physical measure consists of modeling the log-variance as a truncated Brownian semi-stationary process. Then, a deterministic change of measure is applied. The rBergomi model is able to reproduce observed market SP500 smiles with few parameters, but by virtue of the deterministic change of measure, produces flat VIX smiles, in contrast to the upward sloping smiles observed in the market. 
We use the exact solution for a certain inhomogeneous fractional Ornstein-Uhlenbeck equation to build a regime switching stochastic change of measure for the rBergomi model that both yields upward slopping VIX smiles and is equipped with an efficient semi-analytic Monte Carlo method to price VIX options. 
The model also allows an approximation of the VIX, which leads to a significant reduction of the computational cost of pricing VIX options and futures. A variance reduction technique based on the underlying continuous time Markov chain allows us to further reduce the computational cost.
We verify the capabilities of our model by calibrating it to  observed market smiles and discuss the results.
\end{abstract}

\keywords{fractional Ornstein-Uhlenbeck process, rough volatility, VIX option pricing}

\section{Notation}

The notation $L^p$ always means $L^p(\R^d)$ for $d \in \N_1$. If we wish to talk about $L^p$ with respect a specific measurable set $A \subset \R^d$, we write $L^p(A)$. The same applies to $L^p_{loc}$.

We shall denote the physical (or real world) measure by $\mathbb{P}$ and the pricing (or risk neutral) measure by $\mathbb{Q}$. Expected values and Brownian motions are taken with respect to the pricing measure, except when indicated with a $\mathbb{P}$ in superscript. 

The forward variance curve is denoted by
\begin{equation}
\xi_t(u) = \ev{ v_u \mid \F_t}.
\end{equation}
We denote the convolution operator by $\star$, meaning
\begin{equation}
(f \star g )(t) = \int f(s) g(t-s)\,ds.
\end{equation}
If $f, g$ have support in $\R^+$ it becomes
\begin{equation}
(f \star g )(t) = \int_0^t f(s) g(t-s)\,ds.
\end{equation}
For $K\in L^2([0,T])$ and a continuous semi-martingale $dM = bds + \sigma dB$, with $b, \sigma$ locally bounded and adapted, and $B$ a standard Brownian motion (sBm), we may also define the convolution for $t \in [0, T]$ as 
\begin{equation}
(K \star dM)(t) = \int_0^t K(t-s)b_s ds + \int_0^t K(t-s)\sigma_s dB_s.
\end{equation}
We use the notation $\mathcal{E}$ for the stochastic exponential:
\begin{equation}
\mathcal{E}(X)_t = \exp\left(X_t - \frac{1}{X} \qvar{X}_t\right).
\end{equation}

\section{Introduction}

Finding a mathematical model that reproduces the key features of observed market smiles has been a longstanding problem in mathematical finance. To this end, \cite{VolIsRough} have introduced rough volatility models, where the log-variance behaves similarly to a fractional Brownian Motion (fBm).
The rBergomi model, introduced in \cite{PricingRough}, is a rough volatility model that is able to adjust very well to the smiles of the SP500 with a small number of parameters. Moreover, it produces a power-law decaying \textit{at the money skew}, a feature that is not shared by many conventional stochastic volatility models. The rBergomi model is first obtained by modeling the log-variance as a truncated Brownian semi-stationary process (TBSS), which is motivated by the empirical finding that increments of log-volatility behave similarly to those of fBm. Afterwards, a deterministic change of measure is applied, preserving analytic tractability. For evidence for rough volatility see \cite{Alos2007}, \cite{Microstructure}, \cite{Fukasawa2020VolatilityHT} and \cite{Livieri2018RoughVE}. 

Unfortunately, by virtue of the deterministic change of measure, the rBergomi model produces flat smiles for the VIX index. This feature is very inconsistent with the market, where the VIX smile is upward slopping.  In order to circumvent this problem, multiple solutions have been proposed. One possibility is to propose a stochastic volatility of the TBSS (which acts as a stochastic vol-of-vol). This approach was first proposed in \cite{ModulatedVolterra}, where analytic tractability is ensured by assuming the stochastic vol-of-vol is Markovian and independent of the volatility, together with some extra assumptions. This approach is further explored in \cite{LSMC}, where a Least Squares Monte Carlo method is proposed to keep numerical tractability whilst dropping the independence assumption. For a discussion about capturing the VIX skew, see \cite{Alos18}.

One of the challenges of proposing more complex rough volatility models is to preserve analytic tractability. Due to non-Markovianity, classical techniques involving PDE's are not available, and costly Monte Carlo simulations are the only viable alternative. Recently, there has been an attempt to solve this problem by considering rough volatility models of the affine type, from which the most well known is the rough Heston model proposed in \cite{rHeston}. Remarkably, the classical Riccati equation appearing in the characteristic function of the Heston model is replaced by a fractional Riccati equation in the rough Heston model. For details, see \cite{rHeston-Char}. We may also name the \fOU process, where the diffusion term is non-random. The original RFSV model of \cite{VolIsRough} consisted of modeling the log-volatility as a process of the \fOU type. It has also been considered in \cite{Wang-fOU}. 

The general theory of affine Volterra processes is extensively studied in \cite{affine}, where it is shown that the conditional moment generating function (cMGF) can be written in terms of a Riccati-Volterra equation, thereby extending previous results concerning classical affine diffusions (see, for instance, \cite{Filipovi2005TimeinhomogeneousAP}).  
The body of literature concerning affine processes, namely in finance, is vast. See, per example,
\cite{Jaber2020TheCF}, \cite{Jaber2019WeakEA}, \cite{Jaber2019AWS}, \cite{Comte2012AffineFS}, \cite{Wang2008ExistenceAU}, \cite{InhomVolterra}
 and references therein.

In this paper, we propose a \fOU stochastic change of measure for the rBergomi model that produces upward slopping VIX smiles whilst maintaining analytic tractability. The change of measure is obtained by explicitly solving the corresponding fractional affine Volterra equation. This approach has the advantage of giving us a description of market dynamics both under the physical measure $\mathbb{P}$ and the pricing measure $\mathbb{Q}$. 

The paper is organized as follows. In \cref{sec:rBergomi}, we introduce the (generalized) rBergomi model, both under the physical measure and under a general stochastic change of measure. Afterwards, in \cref{sec:fOU}, we discuss the \fOU process and derive the analytical formulae needed for efficient VIX pricing. Next, in \cref{sec:scm}, we propose the stochastic change of measure for the generalized rBergomi model and obtain a semi-closed formula for the forward variance curve. In \cref{sec:control-variate}, we apply the control variate trick inspired by \cite{ModulatedVolterra} and obtain an approximation for the VIX method which significantly reduces the computing time.
In \cref{sec:var-red}, we further reduce computing times by applying a Monte Carlo variance reduction technique through importance sampling of the continuous time Markov chain. 
Then, in \cref{sec:calib}, we discuss model calibration and display the results. The comparison of the various numerical methods considered throughout the paper can be found in \cref{sec:performance}.  Finally, in \cref{sec:conclusion}, we summarize our conclusions and mention further research problems.

\section{The rBergomi Model} \label{sec:rBergomi}

The rBergomi model was introduced in \cite{PricingRough}.
Motivated by empirical data concerning the log-increments of volatility, the authors propose the following model under the physical measure $\mathbb{P}$:
\begin{equation} \label{eq:rB-P}
v_u = A_0(u) 
\exp\left(
2 \sqrt{\gamma} \int_0^u K(u-s) \, d\Wp_s
\right),
\end{equation}
where $K$ is the fractional kernel
\begin{equation}
K(u-s) = (u-s)^{\alpha-1},
\end{equation}
for $\alpha \in (1/2,1)$, and $A_0$ is a deterministic function. We assume zero interest rates for simplicity. Thus, the dynamics of the price process $S_t$ are
\begin{equation}
dS_t = S_t( \zeta_t dt + \sqrt{v_t} d\Bp_t),
\end{equation}
where $\Bp$ is a sBm correlated with $\Wp$ and defined as
\begin{equation}
\Bp = \rho \Wp + \bar{\rho}\bar{\Wp},
\end{equation}
for $\rho \in (-1,1)$ and $\bar{\rho} = \sqrt{1-\rho^2}$.
In order to price options on a fixed time horizon $T>0$, we need to apply a change of measure to the pricing measure $\mathbb{Q}$. There are essentially two components to the change of measure, applied to the independent Brownian motions $W$ and $\bar{W}$. Thus, a general change of measure for $W$ is characterized by a suitable adapted process $\lambda$ such that
\begin{equation} \label{eq:lambda-dW}
dW_t^\mathbb{P} = dW_t + \lambda_t dt.
\end{equation}
Moreover, the change of measure makes $S$ a $\mathbb{Q}$-martingale so that
\begin{equation}
dB_t = d\Bp_t + \frac{\zeta_t}{\sqrt{v_t}} dt.
\end{equation}
Thus, we obtain the (extended) rBergomi model, where the variance is given by
\begin{equation} \label{eq:rB}
v_u = A_0(u) 
\exp\left(
2 \sqrt{\gamma} \int_0^u K(u-s) \, dW_s
+ 2 \sqrt{\gamma} \int_0^u \lambda_s K(u-s) \, ds
\right),
\end{equation} 
and the price satisfies
\begin{equation}\label{eq:rB-price}
dS_t = S_t \sqrt{v_t} dB_t.
\end{equation}

When $\lambda$ is deterministic, it gets absorbed into the deterministic function $A_0$ and gives rise to the original rBergomi model of \cite{PricingRough}. In this case, the variance is log-normal. This means that the VIX (see \cite{VIX-WhitePaper}), which is given by
\begin{equation} \label{eq:VIX}
VIX_t = \sqrt{ 
\frac{1}{\Delta}
\int_t^{t+\Delta}
\xi_t(u)
\, du,
}
\end{equation}
will be also approximately log-normal and lead to flat VIX smiles.

The goal of this paper is to propose a stochastic change of measure $\lambda$ for the rBergomi model which produces upward slopping VIX smiles and at the same type provides an efficient semi-analytic Monte Carlo method to price VIX options. 

To this end, we will first need to make an excursus to the theory of affine Volterra processes, and in particular the \fOU process.

\section{Fractional \OU process} \label{sec:fOU}

\subsection{Homogeneous case}
The $d$-dimensional \fOU process is the solution to the fractional SDE
\begin{equation} \label{eq:fOU-SDE-hom}
X_t = X_0 + \int_0^t   K(t-s) \theta(\mu - X_s) ds + \int_0^t  K(t-s) \sigma dZ_s,
\end{equation}
where $Z$ is a $d$-dimensional sBm, $\mu, \theta \in \R^d$ and $\sigma, \theta \in \R^{d \times d}$ and $\mu \in \R^d$. It is a particular case of the more general class of affine Volterra processes, discussed in \cite{affine}. Although not Markovian, this class of processes still possesses
an exponential-affine conditional moment generating function, as it can be seen in \cite[Theorem~4.3]{affine}. Moreover, we have an explicit expression for the solution.

\subsection{Inhomogeneous \fOU}

The inhomogeneous \fOU process is defined similarly to the above, but the parameters may depend on time:
\begin{equation} \label{eq:fOU-SDE}
X_t = X_0 + \int_0^t  K(t-s) \theta(s)(\mu(s) - X_s)ds +   \int_0^t  K(t-s) \sigma(s) dZ_s.
\end{equation}
The class of inhomogeneous affine processes were recently studied in \cite{InhomVolterra}. By \cite[Theorem~2.1]{InhomVolterra}, we also have an exponential-affine Laplace transform formula for inhomogeneous affine Volterra models with continuous coefficients, under some mild assumptions.
\subsection{Time dependent mean-reversion}
We are interested in the case when only $\mu$ is time  dependent, since it will allow later to remove the log-normality of the variance whilst keeping analytical tractability. Although $\mu$ is neither homogeneous nor continuous, we may still use \cite[Lemma~2.5]{affine} to build an explicit solution. 

In order to ensure $K \star dZ$ has a continuous version, we assume the kernel satisfies \cite[condition~(2.5)]{affine}.

\begin{assumption}\label{ass:K}
Assume $K\in L^2_{loc}(\R_+,\R)$ and there is $\gamma\in(0,2]$ such that $\int_0^h K(t)^2dt = O(h^\gamma)$
and $\int_0^T (K(t+h)-K(t))^2 dt = O(h^\gamma)$ for every $T<\infty$.
\end{assumption}

We also introduce general assumptions on the function $\mu$, that in particular allow for piecewise constant functions.

\begin{assumption}\label{ass:b}
Assume $\mu \in L^q_{loc}$, where $q$ is such that $1/p+1/q=1$, and $K \in L^p_{loc}$. 
\end{assumption}
\begin{rmk}
The fractional kernel satisfies \ref{ass:K}. By \cref{ass:K} we know $K \in L^2_{loc}$, but it may be that $K \in L^p_{loc}$ for $p > 2$, which will yield a weaker condition on $\mu$ (since $q$ will be required to be lower). In the case of the fractional kernel, we know $K \in L^p_{loc}$ for any $p<1/(1-\alpha)$ and in particular for $p=2$ (since $\alpha>1/2$). Thus, at worst, choosing $q=2$ always works.
\end{rmk}

Thus, we have the following theorem.

\begin{thm} \label{thm:fsde-sol}
Suppose $K$ satisfies \cref{ass:K} and $\mu$ satisfies \cref{ass:b}. Let $T>0$. Let $\theta, \sigma \in \R^d$. Let $F:[0,T] \to \R^d$ be a continuous function and $Z$ be a sBm. Denote also by $R_\theta$ the resolvent of second kind (see \cref{dfn:resolvent}) of $K\theta$ and define $E_\theta = K-R_\theta \star K$. Define also, for $0 < c < d$, the deterministic function
\begin{equation}
H_{c,d}(u) = \int_c^d E_\theta(u-s) \theta \mu(s) ds
\end{equation}
and the Volterra process
\begin{equation}
Y_{c,d}^\sigma(u) = \int_c^d E_\theta(u-s) \sigma dZ_s.
\end{equation}

Then the affine Volterra equation
\begin{equation} \label{eq:thm-SDE}
X_t = F(t) + \int_0^t  K(t-s) \theta(\mu(s) - X_s)ds +   \int_0^t  K(t-s) \sigma dZ_s
\end{equation}
has a unique continuous strong solution on $[0,T]$. Moreover, the solution is given explicitly by
\begin{equation}
\label{eq:sol-time-dep}
X_u = g(u) + H_u + Y_u^\sigma ,
\end{equation}
where 
\begin{equation} \label{eq:sol-g}
g(u) = F(u) - \int_0^u R_\theta(u-s)F(s) ds,
\end{equation}
$H_u := H_{0,u}(u)$ is a deterministic function and $Y_u^\sigma := Y_{0,u}^\sigma(u)$ is a (fractional) Gaussian Volterra  process. 
\end{thm}

\begin{proof}
Let $b=\theta \mu$. Consider the function
\begin{equation}
\tilde{F}(t) = F(t) + (K \star b)(t).
\end{equation}
Since the deterministic function $b$ satisfies $b \in L^q_{loc}$ and $K \in L^p_{loc}$, it follows by \cref{lem:cts-conv} that $K \star b$ is continuous. Since $F$ is assumed to be continuous, it follows that $\tilde{F}$ is continuous. Note that \eqref{eq:thm-SDE} can be written as
\begin{equation}
X = \tilde{F} + (-K\theta) \star X + K\star(\sigma dZ).
\end{equation}
Provided $X$ is continuous, by \cite[Lemma~2.5]{affine}, $X$ solves the above if and only if
\begin{equation}
X =  \tilde{F}- R_\theta \star \tilde{F} + E_\theta \star (\sigma dZ).
\end{equation}
Now notice that
\begin{align*}
 \tilde{F}- R_\theta \star \tilde{F} + E_\theta \star (\sigma dZ)
 &= F - R_\theta \star F + (K \star b) - R_\theta \star (K \star b) + E_\theta \star(\sigma dZ) \\
&= g + (K - R_\theta \star K) \star b + E_\theta \star (\sigma dZ) \\
&= g + E_\theta \star b +  E_\theta \star (\sigma dZ),
\end{align*}
where we used the associativity of the convolution operator for deterministic functions in the second equality and the definition of $E_\theta$ in the third equality. Thus, we only have to check that $X$ admits a continuous version. Indeed, by the properties of the resolvent (see \cref{rmk:Gripen}), since $K \in L^2_{loc}$, then also $R_\theta \in L^2_{loc}$. 
Then, since $K$ satisfies \cref{ass:K}, by \cite[Example 2.3.(v)]{affine} it follows that $E_\theta \star K$ also satisfies \cref{ass:K}. The fact that $E_\theta$ satisfies \cref{ass:K} is now a consequence of \cite[Example 2.3.(iii)]{affine}. Thus, by \cite[Lemma~2.4]{affine}, we conclude that $E_\theta \star (\sigma dZ)$ admits a continuous version. Finally, since $\tilde{F}$ is continuous and $R_\theta \in L^p_{loc} \subset L^1_{loc}$, by \cref{lem:cts-conv} we also conclude that $\tilde{F} \star R_\theta$ is continuous.

\end{proof}
We now turn to one dimensional case $d=1$. Using the above theorem it is easy to obtain the exponentially affine formula for the cMGF.
\begin{prop}
Let $w \in \R$ and $X, Y^\sigma, g$ be as in \eqref{eq:sol-time-dep}, with $d=1$. To lighten notation, write $Y = Y^{\sigma=1}$. Then
\begin{equation}
\label{eq:affine-lapace-time-dep}
\ev{ \exp\left(wX_u\right) \mid \F_t} = \exp \left[
w (g(u) + H_u + \sigma Y_{0,t}(u) ) + w^2 \sigma^2e_t(u)
\right]
\end{equation}
where
\begin{equation} \label{eq:def-et}
e_t(u, \sigma) = \frac{1}{2} 
\sigma^2 \int_t^u E_\theta^2(u-s)\, ds.
\end{equation}
\end{prop}
\begin{proof}
It is clear that
\begin{equation}
\ev{ \exp(wX_u) \mid \F_t}  = \exp\left[w (g(u) + H_u) \right] \ev{ \exp(wY_u) \mid \F_t}
.
\end{equation}
Then
\begin{align*}
\ev{ \exp(wY_u) \mid \F_t}&=  \ev{
\exp(wY_{0,t}(u) + wY_{t, u}(u) \mid \F_t 
} \\
&=\exp(wY_{0,t}(u)) \ev{ \exp(wY_{t,u}(u) ) \mid \F_t} \\
&=\exp(wY_{0,t}(u))  \ev{
\exp(w Y_{t,u}(u) )
} \\
&= \exp\left(
wY_{0,t}(u) +  
\frac{1}{2} w^2 \sigma^2 \int_t^u E_\theta^2(u-s)\, ds
\right),\end{align*}
where we used the fact that $Y_{0,t}(u)$ is $\F_t$-measurable and the fact that $\sigma Y_{t,u}(u)$ is Gaussian and independent of $\F_t$ with zero mean, and variance given by $2e_t(u, \sigma)$.

\end{proof}
It is actually possible to express $R_\theta$ and $E_\theta$ in terms of the Mittag-Leffler function, as done in \cref{lem:eB-exact}. Moreover, if we assume $\mu$ is piecewise constant and $F$ is constant, we get explicit expressions for $H$ and $g$.
\begin{prop}\label{prop:Hg-exact}
Let $a < b$. Suppose $\mu$ is piecewise constant so that it can be written as
\begin{equation}
\mu(s) = \sum_{k=0}^n \mu_k \ind{[t_k, t_{k+1} )}(s),
\end{equation}
where $a = t_0 < t_1 < ... < t_n < t_{n+1} = b$.
Then
\begin{equation} \label{eq:h-decomp}
\begin{split}
H_{a, b} (u) &= \sum_{k=0}^n \mu_k I_k, \\
I_k &:= \int_{t_k}^{t_{k+1}} \theta E_\theta(u-s) \, ds.
\end{split}.
\end{equation}
Moreover,
\begin{equation}
I_k = 
E_{\alpha, 1}(-\theta\Gamma(\alpha)(u-t_{k+1})^\alpha) -  E_{\alpha, 1}(-\theta\Gamma(\alpha)(u-t_k)^\alpha) 
.
\end{equation}
Also, if $F \equiv x_0 \in \R$ we have 
\begin{equation}
g(u) = x_0\left( 1 - E_{\alpha, 1}(-\theta \Gamma(\alpha)u^\alpha \right).
\end{equation}
\end{prop}
\begin{proof}

Both facts an easy consequence of the elementary integral equalities involving $R_\theta$ and $E_\theta$, stated and proved in \cref{lem:eB-exact} and \cref{lem:eB-integral}.

\end{proof}

\section{Change of measure via regime switching fractional Ornstein-Uhlenbeck} \label{sec:scm}

We now turn to the case where $\mu$ follows a continuous time Markov chain (CTMC) with $m$ states, independent of the driving Brownian motions.
\begin{prop} \label{prop:laplace-fOU-RS}
Let $\theta$ and $\sigma$ be constant and $\mu$ be independent of $Z$, Markovian and such that it satisfies \cref{ass:b} almost surely. Let $0 < t \leq u \leq T=t+\Delta$. Let $F$ be a continuous function
and $g$ be as in \eqref{eq:sol-g}. For a fixed function $\mu$, let $X^{\{\mu(s) \mid 0 \leq s \leq T\}}$ be as in \eqref{eq:sol-time-dep}.  Now define $X$ by
\begin{equation}
X_u(\omega) = X_u^{ \{\mu_s(\omega) \mid 0 \leq s \leq T \}}(\omega).
\end{equation}
Then
\begin{equation}
\ev{ \exp (wX_u) \mid \F_t} = \exp\left( 
w (g(u) + H_{0,u}(u) + \sigma Y_{0,t}(u) )
+ w^2 e_t(u, \sigma) 
\right)G(w, u-t, \mu_t) ,
\end{equation}
where $G$ is the deterministic function
\begin{equation} \label{eq:def-G}
G(w,\tau,z) = \ev{ \exp \left( w\int_0^\tau \theta E_\theta(\tau-s) \mu(s) \,ds \right) \Big| \mu_0 = z }.
\end{equation}
\end{prop}
\begin{proof}
By the tower property we have
\begin{equation}
\ev{ \exp (wX_u) \mid \F_t} = \ev{ 
\ev{ \exp (wX_u) \mid \F_t \lor \mathcal{G}_T } 
\mid \F_t
},
\end{equation}
where $\mathcal{G}_T$ is the sigma algebra generated by  $(\mu_s)_{0 \leq s \leq T}$.
Since $\mu$ is independent of $Z$, we may apply \eqref{eq:affine-lapace-time-dep}:
\begin{equation}
\ev{ \exp (wX_u) \mid \F_t} = 
\ev{ \exp \left( w (g(u) + H_{0,u}(u) + \sigma Y_{0,t}(u) )+  w^2 e_t(u, \sigma) \right) \mid \F_t}.
\end{equation}
Note that $H_u = H_{0,u}(u)$ is no longer deterministic. By taking out the measurable terms and rearranging
\begin{equation}
\exp\left( w (g(u)+ H_{0,t}(u)+ \sigma Y_{0,t}(u))+ w^2e_t(u, \sigma) \right) \ev{
\exp(wH_{t,u}(u)) \mid \F_t}.
\end{equation}
Finally, the fact that $\mu$ is a CTMC, together with a change of variables implies that
\begin{equation}
\ev{
\exp(wH_{t,u}(u)) \mid \F_t
} = \ev{
\exp(wH_{t,u}(u)) \mid \mu_t
} = G(w,u-t,\mu_t).
\end{equation}
\end{proof}
We may now introduce a stochastic change of measure based on the \fOU process.

\begin{dfn} \label{dfn:smc-fOU}
Consider the rBergomi model as in \cref{sec:rBergomi}, where the variance is given by \eqref{eq:rB}. Introduce the Brownian motions $Z^\mathbb{P}, \bar{Z}^\mathbb{P}$, independent of each other and of $\bar{W}^\mathbb{P}$. Let $\eta \in (-1,1)$, with $\eta'=\sqrt{1-\eta^2}$ and write $\Wp$ as
$$
\Wp_s = \eta Z^\mathbb{P}_s + \bar{\eta} \bar{Z}^\mathbb{P}_s.
$$
Let then $\mu$ follow a CTMC, independent of all the Brownian motions. Let $X$ be as in \cref{prop:laplace-fOU-RS}, where $\sigma=\eta$ and the driving sBm is $Z$. We define the \fOU regime switching change of measure by
\begin{equation} \label{eq:lambda-fOU}
\lambda_s(\omega) = \theta\ (\mu_s(\omega) - X_s(\omega)).
\end{equation}
\end{dfn} 

\begin{rmk}
The change of measure $\lambda$ is no longer deterministic, which means it is not obvious that $W$ is a $\mathbb{Q}$ Brownian motion. To this end, we apply Girsanov's theorem.  For details, we refer to \cref{sec:girsanov}.
\end{rmk}

\begin{rmk}
Note that, since $v$ follows \eqref{eq:rB} and $X$ solves \eqref{eq:thm-SDE}, we may write
\begin{equation}\label{eq:fOU-model}
\log( v_u/A_0(u) ) = 2\sqrt{\gamma} \left(
X_u -F(u) + \bar{\eta} M_u
\right),
\end{equation}
where $M$ is the Riemann-Liouville fBm 
$$
M_u = (K \star d\bar{Z} )(u)= \int_0^uK(u-s) d\bar{Z}_s.
$$
Now it is easy to derive a semi-closed expression for the forward variance curve, which can be used to obtain  the VIX via a standard quadrature method.
\end{rmk}
\begin{prop}
Suppose the assumptions of 
\cref{prop:laplace-fOU-RS} are verified. Let X and $\mu$ be as in \cref{dfn:smc-fOU} with $w = 2\sqrt{\gamma}$. Then we have the following formula for the forward variance curve
\begin{equation}
\label{eq:fvc-formula}
\begin{split}
\xi_t(u) &= \xi_0(u)\frac{
	 G( w,u-t,\mu_t ) 
	}{
	G(w,u,\mu_0)
	}
\exp\left[ 
w \Lambda(0,t,u)
+ w^2 \lambda(0,t,u)
\right],
\end{split}
\end{equation}
where
\begin{equation}
\Lambda(0,t,u) = H_{0,t}(u) + \eta Y_{0,t}(u) 
+ \bar{\eta} M_{0,t}(u),
\end{equation}
\begin{equation}
\lambda(0, t, u) = e_t(u, \eta) -e_0(u, \eta) + m_t(u, \eta) - m_0(u,\eta),
\end{equation}
\begin{equation}
M_{a, b}(u) = \int_a^b K(u-s) \, d\bar{Z}_s,
\end{equation}
and
\begin{equation}
\label{eq:et}
m_t(u, \eta) = \log\ev{ \exp ( \bar{\eta} M_{t, u}(u)) \mid \F_t} =  \frac{1}{2}(1-\eta^2)\frac{(u-t)^{2H}}{2H}.
\end{equation}
\end{prop}
\begin{proof}
Note that, by definition
\begin{equation}
\xi_t(u) = A_0(u) \ev{ v_u/A_0(u) \mid \F_t}.
\end{equation}
Then, by \eqref{eq:fOU-model},
\begin{equation}
\xi_t(u) = A_0(u) \ev{ 
\exp \left(
w (X_u - F(u) + \bar{\eta} M_u)
\right)
\mid \F_t
}.
\end{equation}
Now, using \eqref{eq:sol-time-dep}, observe we can write $X$ as
\begin{equation}
X_u = g(u) + H_{0,t}(u) + \eta Y_{0,t}(u) + H_{t,u}(u) + \eta Y_{t,u}(u).
\end{equation}
Since $Z,\mu$ and $\bar{Z}$ are all independent, it follows that $H_{t,u}(u),Y_{t,u}(u)$ and $M_{t,u}(u)$ are all conditionally independent given $\F_t$. Using also the fact that $M_{0,t}(u), Y_{0,t}(u)$ and $H_{0,t}(u)$ are all $\F_t$-measurable, it follows that $X$ and $M$ are conditionally independent given $\F_t$. Hence
\begin{equation}
\xi_t(u) = A_0(u) e^{-wF(u)} \ev{ \exp(wX_u) \mid \F_t)}
\ev{ \exp(w\bar{\eta}M_u) \mid \F_t}.
\end{equation}
Now apply \cref{prop:laplace-fOU-RS} together with formula \eqref{eq:et} and obtain
\begin{equation}
\label{eq:fvc-formula1}
\begin{split}
\xi_t(u) &= A_0(u) G( w,u-t,\mu_t ) \times \\
&\times \exp\left[
w(g(u)-F(u) + \Lambda(0,t,u)) + w^2 (e_t(u, \eta) + m_t(u, \eta))
\right].
\end{split}
\end{equation}
Finally, if we apply \eqref{eq:fvc-formula1} with $t=0$ we have
\begin{equation}
\xi_0(u) = A_0(u) G(w,u,\mu_0) \exp\left[
w (g(u)-F(u)) + w^2 \left( m_0(u,\eta)+e_0(u, \eta) \right)
\right]
\end{equation}
and hence
\begin{equation} \label{eq:a0-market}
A_0(u) = \frac{\xi_0(u)}{G(0,u,\mu_0)} \exp\left[ -w(g(u)-F(u))- w^2 \left( m_0(u,\eta)+e_0(u, \eta) \right)
\right].
\end{equation}
Plugging this equation into \eqref{eq:fvc-formula1} yields the result.
\end{proof}

\begin{rmk}
The forward variance appearing in \eqref{eq:fvc-formula} can either be extracted from the market (see \cite{PricingRough}), or we can make the simplifying assumption of a flat forward variance curve and treat $\xi_0(u) \equiv \xi_0 \in \R^+$ as a model parameter to be specified.
\end{rmk}

\begin{rmk}
We can express the forward variance curve in our model in terms of the classical rBergomi forward variance curve $\xi^{RB}$. In the rBergomi model, for a given initial curve $\xi_0$ and $v=2\sqrt{\gamma}$, the forward variance curve will be given by
\begin{equation}
\log \xi_t^{RB}(u \mid v) := \log \xi_0(u) + vM_{0,t}(u) + v^2(m_t(u, 0) - m_0(u,0))
.
\end{equation}
For our model, we can express the forward variance curve in terms of the rBergomi forward variance curve as follows
\begin{equation}
\log \xi_t(u) = \log \xi_t^{RB}(u \mid \bar{\eta} w) + \Xi,
\end{equation}
where
\begin{equation}
\Xi = \log \frac{G(w, u-t, b_t)}{G(w, u, b_0)} + w(H_{0,t}(u) + \eta Y_{0,t}(u)) + w^2(e_t(u) - e_0(u)).
\end{equation}
We can verify that, as it had to be, should we parameterize our model so that the change of measure is deterministic, i.e. $\eta=0$ and $\mu$ deterministic, then we recover the rBergomi model. Indeed, since $\eta=0$, we have $e_t \equiv e_0 \equiv Y_{0,t} \equiv  0$. Moreover, since $H$ is deterministic,
\begin{align*}
\log \frac{G(w, u-t, \mu_t) }{G(w, u, \mu_0)} &= \log \ev{
\exp (wH_{t,u}(u))  \mid \F_t
}  - \log \ev{ \exp (wH_{0,u}(u)) } \\
&= -wH_{0,t}(u).
\end{align*}
This implies $\Xi = 0$. Finally, since $\eta=0$, we have $\bar{\eta}=1$ and thus
\begin{equation}
\log \xi_t(u) = \log \xi_t^{RB}(u \mid w),
\end{equation}
as expected.
\end{rmk}

\begin{rmk}
If $\mu$ were deterministic, the change of measure would amount to changing the kernel of the Volterra process, and the variance would still be log-normal, leading to an approximately flat VIX smile.
\end{rmk}

\section{Control Variate} \label{sec:control-variate}

As done in \cite{ModulatedVolterra}, we may use the fact that $\xi_t(u)$ is conditionally log-normal to approximate the VIX and use the approximation as a control variate. Namely, we approximate the integral over a family of conditionally log-normal random variables by the exponential of the integral of their logarithms.

By \cref{prop:laplace-fOU-RS}, we know that, conditional on $\mu(s)_{ 0 \leq s \leq t}$, $\log \xi_t(u)$ is Gaussian. This allows us to conclude that
\begin{equation}
N_t := 
\frac{1}{\Delta}
\int_t^{t+\Delta} \log \xi_t(u) \, du
\end{equation}
is also Gaussian. Let us denote its mean by $\mu_N$ and its variance by $\sigma_N^2$. As a consequence of \cref{prop:laplace-fOU-RS}, 
\begin{equation}
\log \xi_t(u) = \bar{m}(u) + 
w( \eta  Y_{0,t}(u) + \bar{\eta} M_{0,t}(u)),
\end{equation}
where
\begin{equation}
\bar{m}(u) = \log \xi_0(u) +\log G(w, u-t, \mu_t) - \log G(w, u, \mu_0) + w^2 \lambda(0, t, u) + wH_{0,t}(u).
\end{equation}
Since $Y$ and $M$ do not depend on $\mu$ and have zero mean, the conditional mean of $\log \xi_t(u)$ is given by $\bar{m}(u)$. Hence, by Fubini's theorem, the conditional mean of $N_t$ is 
\begin{equation}
\mu_N = 
\frac{1}{\Delta}
\int_t^{t+\Delta} \bar{m}(u) \, du.
\end{equation}
Since $Y, M$ are independent given $\mu$, the conditional variance of $N$ satisfies
\begin{equation}
\sigma^2_N = 
\frac{w^2}{\Delta^2}
 \left[ \eta^2 \sigma^2_Y  + (1-\eta^2) \sigma^2_M \right],
\end{equation}
where
\begin{equation}
\sigma^2_Y = \var \left[ \int_t^{t+\Delta} \left( \int_0^t E_\theta(u-s) \, dZ_s \right)  du  \right]
\end{equation}
and
\begin{equation}
\sigma^2_M = \var \left[ \int_t^{t+\Delta} \left( \int_0^t K(u-s) \, d\bar{Z}_s \right) du \right].
\end{equation}

By the stochastic Fubini theorem it follows that
\begin{equation}
\sigma^2_Y = \int_0^t  \left(\int_t^{t+\Delta} E_\theta(u-s) du \right)^2 ds
\end{equation}
and
\begin{equation}
\sigma^2_M = \int_0^t  \left(\int_t^{t+\Delta} K(u-s) du \right)^2 ds.
\end{equation}
In particular, $\sigma_N^2$ does not depend on $\mu$.

We may use \cref{lem:eB-integral} to make the expression for $\sigma^2_Y$ more explicit:
\begin{equation}
\sigma^2_Y  = \frac{1}{\theta^2}  \int_0^t 
\left[
E_{\alpha,1}(-c(t-s)^\alpha) - E_{\alpha,1}(-c(t-s+\Delta)^\alpha)
\right]^2
\, ds,
\end{equation}
where $c = \theta \Gamma(\alpha)$.
Likewise,
\begin{align*}
\sigma^2_M &= \int_0^t  \left(\int_t^{t+\Delta} K(u-s) du \right)^2 du \\
&= \frac{1}{\alpha^2} \left[
\int_0^t (t+\Delta-s)^{2\alpha} + (t-s)^{2\alpha} - 2 (t+\Delta-s)^\alpha(t-s)^\alpha \, ds
\right] \\
&=  \frac{1}{\alpha^2}  \left[
\frac{1}{2\alpha+1} \left(
(t+\Delta)^{2\alpha+1} - \Delta^{2\alpha+1}
 + t^{2\alpha+1}
\right)
- 2 \int_0^t (x^2+x\Delta)^\alpha \, dx
\right].
\end{align*}
Note that, since $\alpha \in (1/2,1)$, the computation of the  above integrals does not involve singularities.

Using the above, we may obtain the cMGF:
\begin{equation}
\ev{ \exp(z N_t) \mid \mu} = \exp\left( z \mu_N + \frac{1}{2}z^2 \sigma^2_N \right).
\end{equation}
Thus, the logarithm of the cMGF is given by
\begin{align*}
\log \ev{ \exp(z N_t) \mid \mu_t} &= 
 \frac{z}{\Delta} \left( 
\int_t^{t+\Delta} \log \xi_0(u) - \log G(w, u, \mu_0) + w^2 \lambda(0,t,u) \, du
\right) +\\
&+ \frac{1}{2} z^2 \sigma_N^2 + \log \ev{ \exp\left( 
\frac{z}{\Delta}
 \int_t^{t+\Delta} \log G(w, u-t, \mu_t) + wH_{0,t}(u) \, du
\right) }.
\end{align*}
If we apply the approximation
\begin{equation}
VIX_t^2 = \frac{1}{\Delta} \int_t^{t+\Delta} \xi_t(u) \, du \approx \exp\left(
N_t
\right),
\end{equation}
we obtain an approximation of the VIX Future
\begin{equation}  \label{eq:cv-future}
\ev{VIX_t} \approx \ev{ \exp\left(\frac{1}{2} N_t\right)}.
\end{equation}
For a call option on the VIX, we have the approximation
\begin{equation} \label{eq:cv-option}
\ev{ 
\ev{
\left( \exp(N_t/2) - K\right)^+
\mid \mu}
} = \ev{
BS(\mu_N/2, \sigma_N^2/4)
},
\end{equation}
where $BS$ is obtained using the Black-Scholes formula:
\begin{equation}
BS(\mu, \sigma^2) = \mathcal{N}(d_+)F - \mathcal{N}(d_-) K,
\end{equation}
where
\begin{equation}
d_{\pm} = \frac{1}{\sigma} \left(
\log \frac{F}{K} \pm \frac{1}{2}\sigma^2
\right)
\end{equation}
and 
\begin{equation}
F = \exp\left(
\mu + \frac{1}{2} \sigma^2
\right).
\end{equation}

\begin{rmk}
In the rBergomi model, $\mu$ is deterministic and thus we obtain a price given (approximately) by the Black-Scholes formula, leading to a flat VIX smile. We can thus see the VIX option price in our model as (approximately) a weighted average of the rBergomi prices, where the weights are determined by the parameters of the CTMC $\mu$.
\end{rmk}

\section{Variance Reduction via Importance Sampling} \label{sec:var-red}

The distribution of dwelling times for the more extreme values of the CTMC can be deeply unbalanced: only a small percentage of generated paths will contain a significant dwelling time for the more extreme states. For out-of-the-money (OTM) options, both for the SP500 and the VIX, these paths have a significant contribution on the mean, since they will generate non-zero payouts for call options. Thus, it is natural to employ a Monte Carlo Variance Reduction technique (MCVR) via importance sampling.

In general, if we are interested in estimating $\ev{f(X)}$ for a measurable function $f$ and a random variable $X$ distributed according to a density $p$, the simple Monte Carlo approach is to compute
\begin{equation} \label{eq:simple-mc}
\frac{1}{n} \sum_{i=1}^n f(X_i) , X_i \sim p \, \forall i=1, ..., n.
\end{equation}
Should we consider an alternative density $h$ with the intent of reducing the variance of the estimator, we may sample $Y \sim h$ and weight the samples by the likelihood ratio $p/h$:
\begin{equation}
\frac{1}{n} \sum_{i=1}^n f(Y_i) \frac{p(Y_i)}{h(Y_i)} , Y_i \sim h \, \forall i=1, ..., n,
\end{equation}
producing an unbiased estimator for $\ev{f(X)}$. If $p=h$ we of course recover \eqref{eq:simple-mc}. Thus, in order to perform the MCVR, we need access to the densities $p$ and $h$.

Let us consider our case of the CTMC with $m$ states starting at state $s_0$. Let $\{q_i\}_{i=1}^n$ be the jump intensities associated with each state and $p_{i,j}$ the probability of switching from state $i$ to state $j$. Each path of the CTMC with $k$ jumps is identified by the sequence of states it attains $s=(s_0, s_1, ..., s_k)$ and dwelling times $t=(t_0, ..., t_{k-1} )$. By successive conditioning and using the Markov property it is easy to see that the density is given by
\begin{equation}
p(s, t) = e^{-q_m t_m} \prod_{i=0}^{k-1} p_{s_i, s_{i+1}} q_{s_i} e^{-q_{s_i}t_i} \ind{ A_k },
\end{equation}
where
\begin{equation}
t_k = T - \sum_{i=0}^{k-1} t_i
\end{equation}
and
\begin{equation}
A_k = \left\lbrace 
(t_0, t_1,..., t_{k-1} ) \in \R^k :
\sum_{i=0}^{k-1}t_i < T
 \right\rbrace.
\end{equation}

In order to reduce the variance of the estimator, we may simply use a uniform distribution for the dwelling times, which is given by the inverse of the volume of the region under a $(k-1)$-simplex:
\begin{equation}
\frac{1}{h\st} = \int_{ A_k } 1 = \frac{T^k}{k!}.
\end{equation}

Thus, for a function of interest $f=f\st$, our MCVR estimator is given by
\begin{equation} \label{eq:vr-estimator-initial}
\sum_{k=0}^{+\infty} \frac{T^k}{k!} 
\sum_{s \in J(k)}  \frac{1}{n}\sum_{i=1}^n f\sti p\sti,
\end{equation}
where $J(k)$ is the set of all possible state sequences starting at $s_0$ with $k$ jumps, and $\textbf{t} \sim \text{Uniform}(A_k)$, i.e. it is uniformly distributed under the $(k-1)$-simplex. Such a sample can be obtained by sampling from the $k$-dimensional flat Dirichlet distribution (thus obtaining a uniform sample on the $k$-simplex) and removing the last coordinate. 

Next, there are two adjustments that still have to be made to \eqref{eq:vr-estimator-initial}. First, we do not consider the number of jumps up to infinity, but rather truncate the series in by an adequate maximum number of jumps $K_M$, which is chosen based on the model parameters. Secondly, using the same number of samples for every number of jumps is sub-optimal since it dramatically increases the computational cost. For this reason, we applied a stratified sample approach, where the sample size for each number of jumps $n_k$ is in general proportional to its probability mass, but at the same time each number of jumps is required to have a minimum sample size. This allows us to substantially reduce the variance with only a small increase of computational cost. Thus, the MCVR estimator is given by
\begin{equation}\label{eq:vr-estimator}
\sum_{k=0}^{K_M} \frac{T^k}{k!} 
\sum_{s \in J(k)} \frac{1}{n_k}\sum_{i=1}^{n_k} f\sti p\sti.
\end{equation}

Before we dwell into the application of the MCVR to pricing VIX and SP500 options, we note that it can be used to compute \eqref{eq:def-G} for each $\tau$, which is needed for both the VIX and the SP500 options valuation, with 
\begin{equation}
f_\tau\st = \exp\left(
w \int_0^\tau \theta E_\theta(\tau - u ) \mu\st(u) \, du \right)
.
\end{equation}

\subsection{VIX} \label{sec:MCVR-VIX}

Let us consider the approximations for the VIX future and option given by \eqref{eq:cv-future} and \eqref{eq:cv-option}, respectively. Note that the randomness in these approximations depends only on the CTMC $\mu$. The MCVR estimator for the VIX future is then given by \eqref{eq:vr-estimator}
with
\begin{equation}
f\st = \exp \left(
\frac{1}{2} \mu_N\st + \frac{1}{8} \sigma^2_N
\right),
\end{equation}
with $\mu_N$ and $\sigma^2_N$ as in \cref{sec:control-variate}. Note that $\mu_N$ is a deterministic function of the path of the CTMC and $\sigma^2_N$ is deterministic.

In a similar way, the MCVR estimator for the VIX option is built using
\begin{equation}
f\st = BS( \mu_N\st/2, \sigma^2_N/4).
\end{equation}

\subsection{SP500} \label{sec:MCVR-sp}

To price SP500 options, the payout function will also depend on the Brownian paths. But even in this scenario, we may apply the MCVR. We start by considering the Brownian increments of the sBm's $Z, \bar{Z}$ and $B$, which we will denote by $\textbf{W}$. Note that $\textbf{W}$ is independent of $\st$. In the MCVR, we do not consider an alternative density for $\textbf{W}$ but only for $\st$. By independence we obtain the following MCVR estimator of $\ev{f(\textbf{s}, \textbf{t}, \textbf{W})}$
\begin{equation} 
 \sum_{k=0}^{K_M} \frac{T^k}{k!} 
\sum_{s \in J(k)} \frac{1}{n_k}\sum_{i=1}^{n_k} 
f(\textbf{s}, \textbf{t}_i, \textbf{W}_i) 
 p\sti.
\end{equation}

In the case of the SP500 call option with maturity $t$ and strike $K$, the function $f$ is simply given by 
\begin{equation}
f(\textbf{s}, \textbf{t}, \textbf{W}) = \left(
S_t( \textbf{s}, \textbf{t}, \textbf{W}) -K
\right)^+,
\end{equation}
where we use the fact that the Brownian increments, together with the path of the CTMC, totally determine the (discretization of the) underlying price path.

The results in \cref{sec:performance} show that the generated smiles are, on average, very similar, but the variance (and hence the computational cost) reduction  is very significant for the VIX when applying MCVR.

\section{Numerical implementation} \label{sec:numerical}

All numerical simulations are performed on a Linux machine (Ubuntu 20.4), with an AMD Ryzen 7 3800X CPU with 16 threads. The Python language was used, resorting to \textit{numpy} whenever possible for C-like speed. We also made use of the \textit{numba} package, which allows for just in time compilation of code. This is especially useful for code that involves long/nested loops over nested lists of varying length, such as the ones that commonly appear in code related  to CTMC's. These kinds of routines are not easily vectorized in \textit{numpy} and using \textit{numba} allows for much faster execution time without writing convoluted \textit{numpy} code.

\subsection{Pricing SP500 options} \label{sec:pricing-SP500}

To price SP500 options, we use the MCVR method of \cref{sec:MCVR-sp}. Once the variance process is computed, the price process \eqref{eq:rB-price} is easily obtained using a standard Riemann sum approximation.

In order to simulate the variance, we use \eqref{eq:fOU-model}. To simulate $X$, we use \eqref{eq:sol-time-dep}. The process $M$ is a Riemann-Liouville fBm, so we may use the hybrid scheme of \cite{Hybrid2017}. 

By \cref{lem:eB-exact}, $E_\theta=t^{\alpha-1}\psi(t)$, where $\psi$ is a continuous function expressed in terms of the Mittag-Leffler function. Thus, the process $Y$ can also be computed with the hybrid scheme, setting $L_g = \psi$ in the notation of \cite[Remark 3.1]{Hybrid2017}. For the computation of the \ML{}, we used the library provided by \cite{ML-Python}. The function $g$ is known exactly. The process $H$ can be approximated using \cref{prop:Hg-exact}.

Using the \textit{numba} package significantly improved the simplicity and performance of the simulation of $H$, due to the fact that the number of jumps is potentially different for each path.

Finally, the initial curve $A_0$ can be obtained from $\xi_0$ via \eqref{eq:a0-market}.

\subsection{Pricing VIX options}

\subsubsection{Simple Monte Carlo}

Pricing VIX options via simple Monte Carlo amounts to simulating the forward variance curve, for which we resort to \eqref{eq:fvc-formula}. 

The terms $Y_{0,t}(u)$ and $M_{0,t}(u)$ are Volterra Gaussian processes with non-singular kernels, and thus can be approximated using the standard Riemann sum method. The term $H_{0,t}(u)$ can be approximated using \cref{prop:Hg-exact}. The function $m_t$ poses no difficulty, since it has a closed expression. For the function $e_t$, which will involve an integral with a weakly singular integrand, we may use the procedure of \cref{sec:num-integral}.

Finally, we need to compute $G(w, u-t, b^{(i)}_t)$, for each simulation $i$ and $u \in [t, t+\Delta]$. Fortunately, $\mu_t$ can only attain a finite number of values. Thus, we may easily perform a simple Monte Carlo simulation for each of the $m$ possible values of $\mu_t$. That is, we may create $M$ paths of $\mu$ conditioned on $\mu_0=z_k$ for each possible state ${z_1,z_2, ..., z_m}$ and then approximate
$$
G(w, u-t,z_k) \approx \frac{1}{M} \sum_{j=1}^M \exp \left( 
w \int_0^{u-t} \theta E_\theta(u-t-s)\mu^{(j)}(s) \,ds\right).
$$
Again, we use \cref{prop:Hg-exact} to compute the integral.
\begin{rmk}
The pricing of VIX options only requires the Brownian increments of $Z, \bar{Z}$ and $B$, but does not require simulating the processes $H, Y$ and $M$. Recall these processes are computationally more expensive to simulate than $H_{0,t}(u), Y_{0,t}(u)$ and $M_{0,t}(u)$ because of the presence of singularities in the kernels, which requires the hybrid scheme in the case of $Y$ and $M$.
\end{rmk}
\subsubsection{Control Variate and Variance Reduction}
We may use the approximations of \cref{sec:control-variate} as a control variate to reduce the variance of the Monte Carlo estimator, or we may see it as an approximation and drop the simple Monte Carlo method altogether. In practice, it is much more efficient to use it as an approximation since the correlation between the control variate and the Monte Carlo estimator is very high. Finally, the variance reduction techniques described in \cref{sec:MCVR-VIX} let us further reduce the variance of the estimator.
\FloatBarrier
\section{Calibration} \label{sec:calib}

We calibrate our model to the VIX and SP500 smiles separately and then perform a joint calibration. Due to much faster computing times, we used the MCVR estimators of \cref{sec:var-red} for both the SP500 and VIX Smiles. We found the maximum number of jumps $K_M = 4$ to be adequate for the purposes of this paper.

We used the \textit{LMFIT} library for the calibration routines (see \cite{LMFIT}). We performed some exploratory tests and found the \textit{Levenberg-Marquardt} and \textit{Trust Constraint} methods to be particularly effective. The parameter space is simply a multidimensional rectangle. We fixed $\rho=-0.95$, $x_0=0$ and $s_0=0$ in all calibrations, in order to simplify the calibration procedure. 

We tested the calibration with both $m=3$ and $m=2$ states for the CTMC. Since $m=3$ requires 6 more parameters than $m=2$, the calibration was both slower and displayed greater dependence on the initial condition. Moreover, the quality of the fits did not substantively differ, and it was frequent for $\mu_2$ and $\mu_3$ to be close, indicating that perhaps the model was over-parameterized with $m=3$ states. Devising  a robust and efficient calibration procedure with $m=3$ states that substantially improves the quality of the fits is an endeavor that we leave for further research.

All smiles were extracted from Yahoo finance and date from $19$ January 2021. The implied volatilities were computed from the average of the the bid and ask prices.
\FloatBarrier
\subsection{Calibration to SP500 options}
Although the original rBergomi model is a particular case of our \fOU model, and for this reason our model inherits the ability to fit SP500 smiles, it is still interesting to study the calibration to the SP500 under the general \fOU model we propose. First, because potentially multiple minima exist, and thus there can be parameterizations of the model that are not the rBergomi model. Secondly, because the rBergomi models corresponds to choices of parameters in the boundary of the parameter space (zero intensity for jumps, zero mean reversion value or setting all states of CTMC to the same value), making it hard for the calibration procedure to converge to the rBergomi model. Testing the ability of the calibration procedure and our proposed model to adjust to SP500 smiles adds confidence in both the reliability of the calibration procedure and the flexibility of the model.
\begin{table}[H]
\centering
\begin{tabular}{||c c c c||} 
 \hline
 \textbf{Parameter} & \textbf{Min} & \textbf{Max} & \textbf{Calibrated} \\ [0.5ex] 
 \hline\hline
$H$ & 0.07 & 0.13 & 0.0846 \\
$\rho$ & \textit{fixed} & \textit{fixed} & -0.95 \\
$\eta$ & -0.99 & 0.99 & -0.3021 \\
$\theta$ & 0.1 & 10.0 & 1.6672 \\
$\gamma$ & 0.0 & 0.4 & 0.3367 \\
$\mu_1$ & 0.0 & 1.0 & 0.0005 \\
$\mu_2$ & 0.0 & 20.0 & 16.0288 \\
$q_1$ & 0.0 & 2.0 & 0.0193 \\
$q_2$ & 0.0 & 15.0 & 14.4128 \\
$\xi_0$ & 0.0001 & 0.25 & 0.0553 \\
$x_0$ & \textit{fixed} & \textit{fixed} & 0.0 \\
 \hline
\end{tabular}
\caption{SP500 Calibration -- Calibrated Values}
\label{sp-table}
\end{table}
In \cref{sp-table}, we present the calibration results and parameter space for the SP500 calibration. As expected, the model did not converge to the rBergomi model, which lies in the boundary of the parameter region space. In \cref{sp-smile}, we observe the calibrated smile provides a reasonable fit to the observed market smile, staying mostly inside the bid-ask spreads.
\begin{figure}
\centering
\includegraphics[width=345 pt]{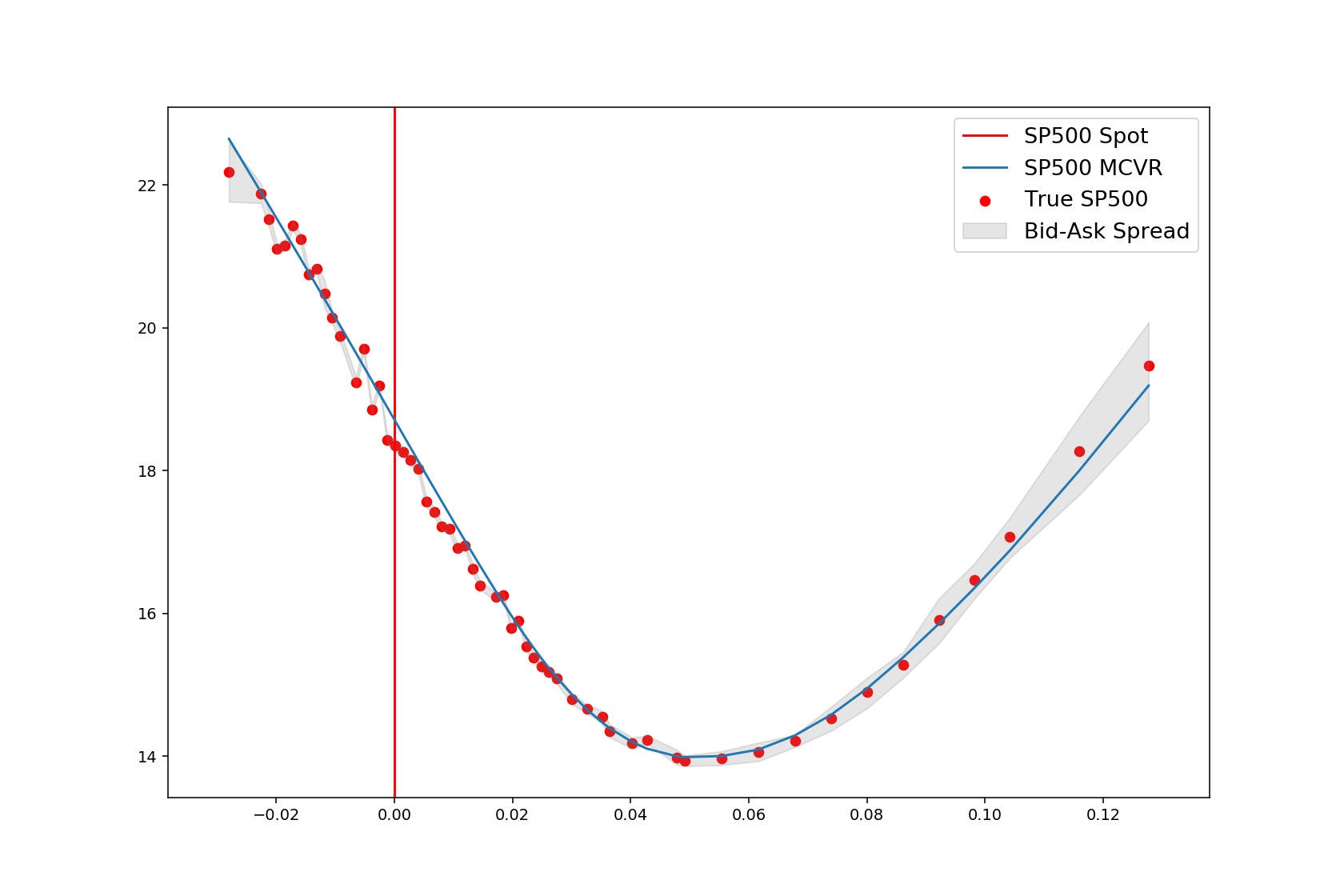}
\caption{Calibrated SP500 Smile}
\label{sp-smile}
\end{figure}
\FloatBarrier
\subsection{Calibration to VIX options and future}
For the calibration to VIX smile, our loss function combined the error with respect to both VIX future and VIX options. The Control Variate approximation described in \cref{sec:control-variate} for VIX options and the VIX future only depends on $\eta^2$. Since we used the MCVR method of \cref{sec:MCVR-VIX}, which inherits the properties of the Control Variate, we limited $\eta \in (0,0.99)$ in the parameter space for the VIX calibration. In \cref{calib-vix-table}, we see the calibrated values.
\begin{table}
\centering
\begin{tabular}{||c c c c||} 
 \hline
 \textbf{Parameter} & \textbf{Min} & \textbf{Max} & \textbf{Calibrated} \\ [0.5ex] 
 \hline\hline
$H$ & 0.07 & 0.13 & 0.0938 \\
$\rho$ & \textit{fixed} & \textit{fixed} & -0.95 \\
$\eta$ & 0.0 & 0.99 & 0.1373 \\
$\theta$ & 0.1 & 10.0 & 5.9165 \\
$\gamma$ & 0.01 & 0.2 & 0.1751 \\
$\mu_1$ & 0.0 & 1.0 & 0.1239 \\
$\mu_2$ & 0.0 & 20.0 & 4.8671 \\
$q_1$ & 0.0 & 2.0 & 0.699 \\
$q_2$ & 0.0 & 15.0 & 13.4365 \\
$\xi_0$ & 0.0001 & 0.25 & 0.0654 \\
$x_0$ & \textit{fixed} & \textit{fixed} & 0.0 \\
 \hline
\end{tabular}
\caption{VIX Calibration -- Calibrated Values}
\label{calib-vix-table}
\end{table}
\begin{figure}[H]
\centering
\includegraphics[width=345 pt]{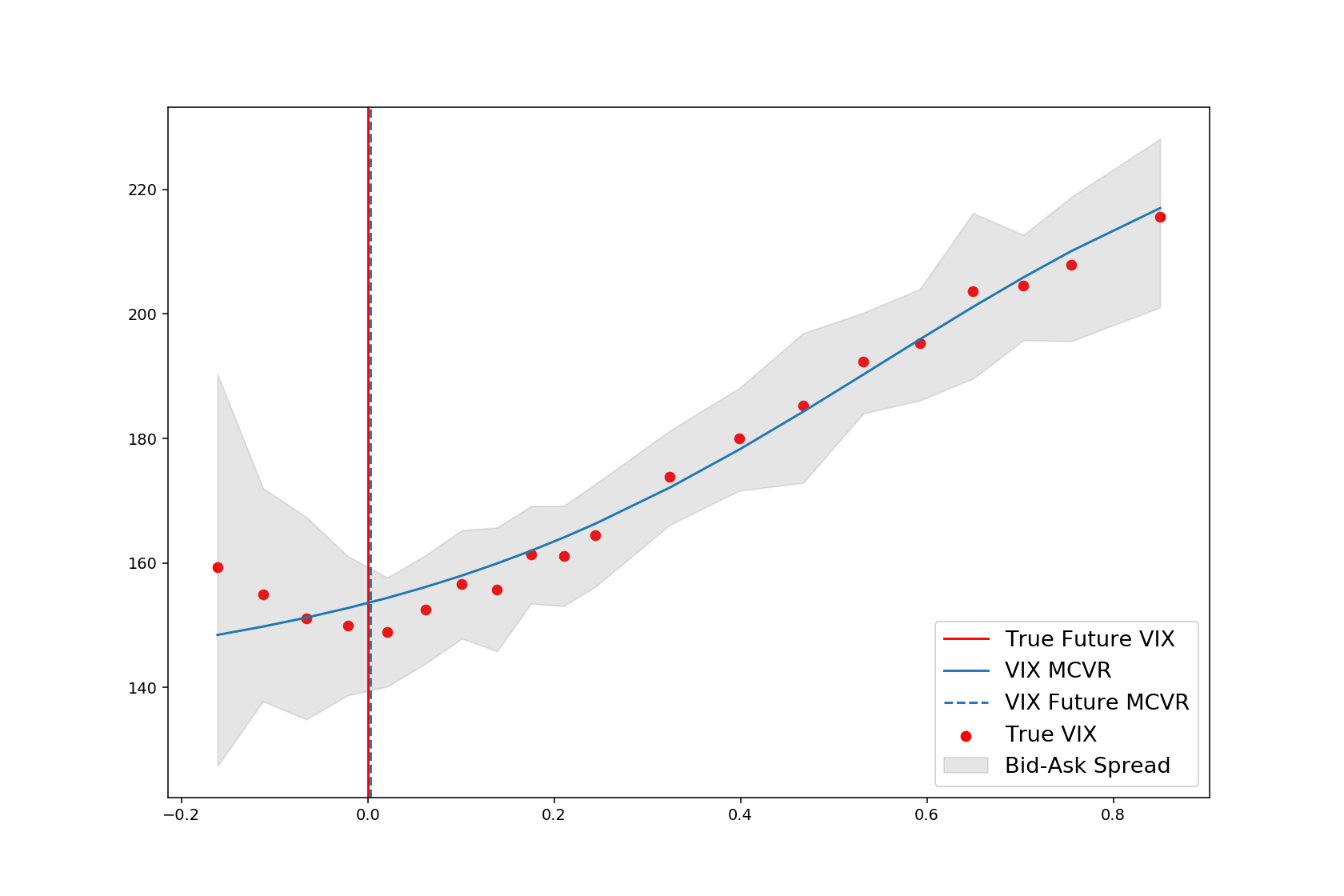}
\caption{Calibrated VIX Smile}
\label{vix-smile}
\end{figure}
\FloatBarrier
\subsection{Joint calibration to both SP500 and VIX options}
The SP500-VIX Joint calibration has been called \textit{the holy grail of volatility modeling} (see \cite{Joint-VIX-rHeston}). One of the main challenges in  rBergomi-like models is that the large SP500 observed in the market will require for the vol-of-vol parameter $\gamma$ to be very large, which would produce VIX implied volatilities of a much higher magnitude than we observe. The extra flexibility provided by our regime switching change of measure allows the model to reproduce the large SP500 skew whist keeping VIX implied volatilities close to the ones observed in the market.
\begin{table}
\centering
\begin{tabular}{||c c c c||} 
 \hline
 \textbf{Parameter} & \textbf{Min} & \textbf{Max} & \textbf{Calibrated} \\ [0.5ex] 
 \hline\hline
$H$ & 0.07 & 0.13 & 0.114 \\
$\rho$ & \textit{fixed} & \textit{fixed} & -0.95 \\
$\eta$ & -0.99 & 0.99 & -0.3792 \\
$\theta$ & 0.0 & 6.0 & 5.6312 \\
$\gamma$ & 0.0 & 0.3 & 0.2468 \\
$\mu_1$ & 0.0 & 5.0 & 1.004 \\
$\mu_2$ & 0.0 & 20.0 & 6.7563 \\
$q_1$ & 0.0 & 2.0 & 0.2821 \\
$q_2$ & 0.0 & 15.0 & 10.1285 \\
$\xi_0$ & 0.0001 & 0.25 & 0.0462 \\
$x_0$ & \textit{fixed} & \textit{fixed} & 0.0 \\
 \hline
\end{tabular}
\caption{Joint Calibration -- Calibrated Values}
\label{calib-joint-table}
\end{table}
In \cref{joint-smile}, we see that the joint calibration is not as good as individual SP500 or VIX calibrations. Nevertheless, it provides a good approximation, which can potentially be improved in future works. %
\begin{figure}
\centering
\includegraphics[width=345 pt]{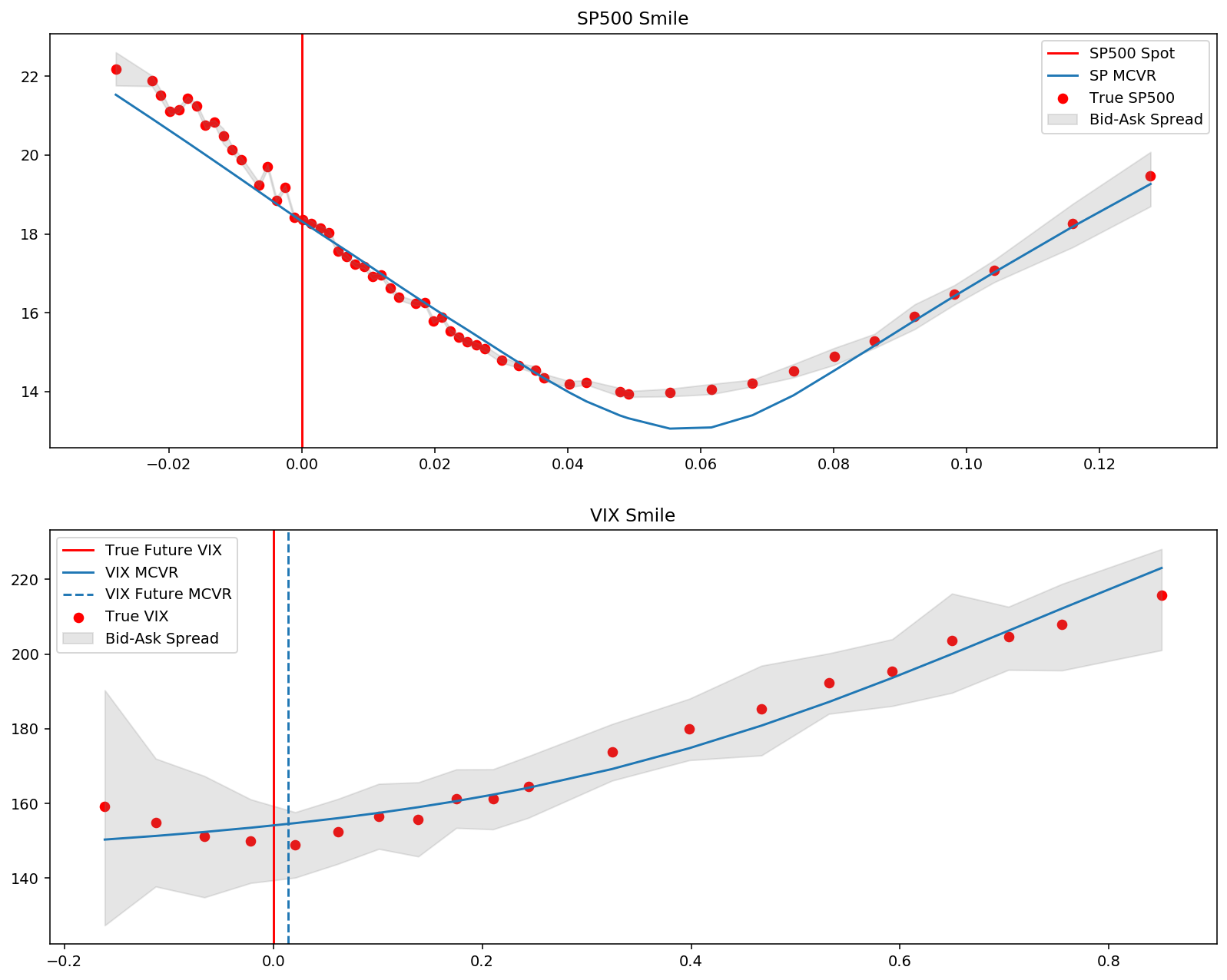}
\caption{Calibrated Joint SP500-VIX Smile}
\label{joint-smile}
\end{figure}
\FloatBarrier
\section{Performance of Numerical Methods} \label{sec:performance}
In order to evaluate the performance of the MCVR, we generate multiple smiles ($1000$ for the SP500 and $2000$ for the VIX) with approximately the same computational budget, using the Monte Carlo Variance Reduction and Simple Monte Carlo methods. In the case of the VIX, we also present the comparison with the Control Variate approximation method. We compute the mean smile across the entire sample, in order to validate that the methods are equivalent in the sense that they produce approximately the same results. We also compute the standard deviation, to compare the performance of the methods. We used the calibrated values found in \cref{sec:calib}.
\FloatBarrier
\subsection{SP500}
For the SP500, as it can be seen in \cref{sp-VarRed-stats}, the generated smiles are indeed very close, As expected, the variance reduction method did not reduce the variance for ITM options but only for OTM options. This is due to the fact that the SP500 smile depends on the generated Brownian increments, which are not affected by the variance reduction method. This behavior is also displayed in \cref{sp-VarRed-errors}. 
\begin{figure}
\centering
\includegraphics[width=345 pt]{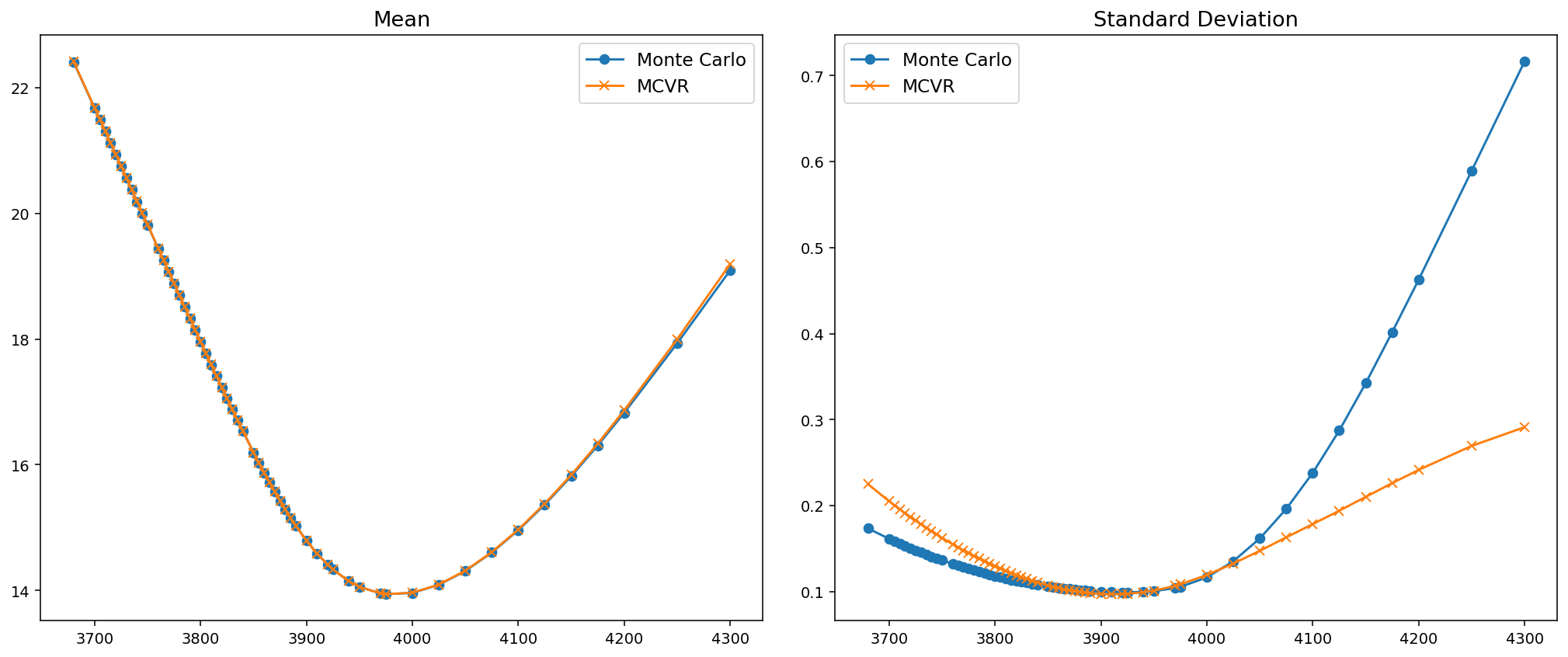}
\caption{Statistics comparing the various numerical methods - SP500}
\label{sp-VarRed-stats}
\end{figure}
\begin{figure}
\centering
\includegraphics[width=345 pt]{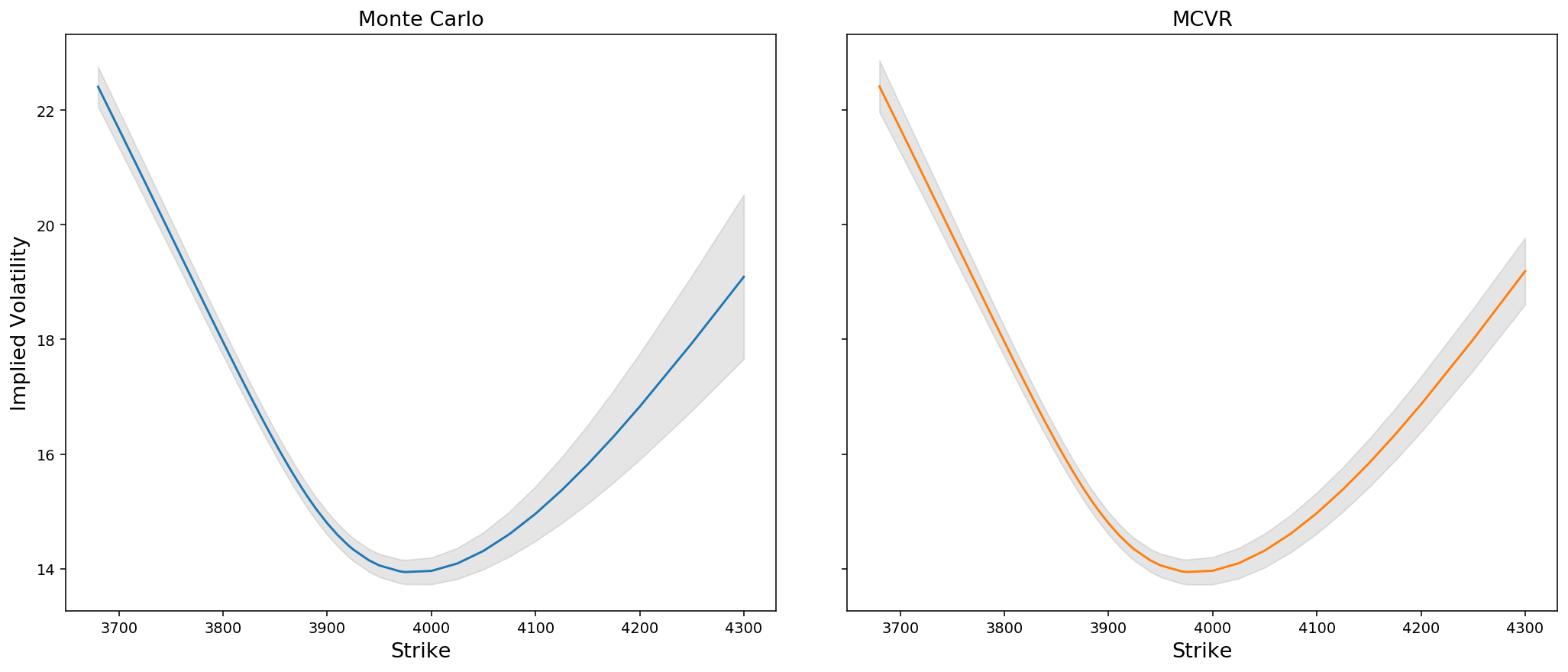}
\caption{Error bars for various numerical methods (2 Standard Deviations) - SP500}
\label{sp-VarRed-errors}
\end{figure}
We also note that the impact of the MCVR depends on the chosen model parameters. We performed a SP500 calibration with $m=3$ states, and generated SP500 smiles using the calibrated parameters. In this case, the MCVR significantly reduced the variance, as it can be seen in \cref{sp-VarRed-stats-3s} and \cref{sp-VarRed-errors-3s}.
\begin{figure}
\centering
\includegraphics[width=345 pt]{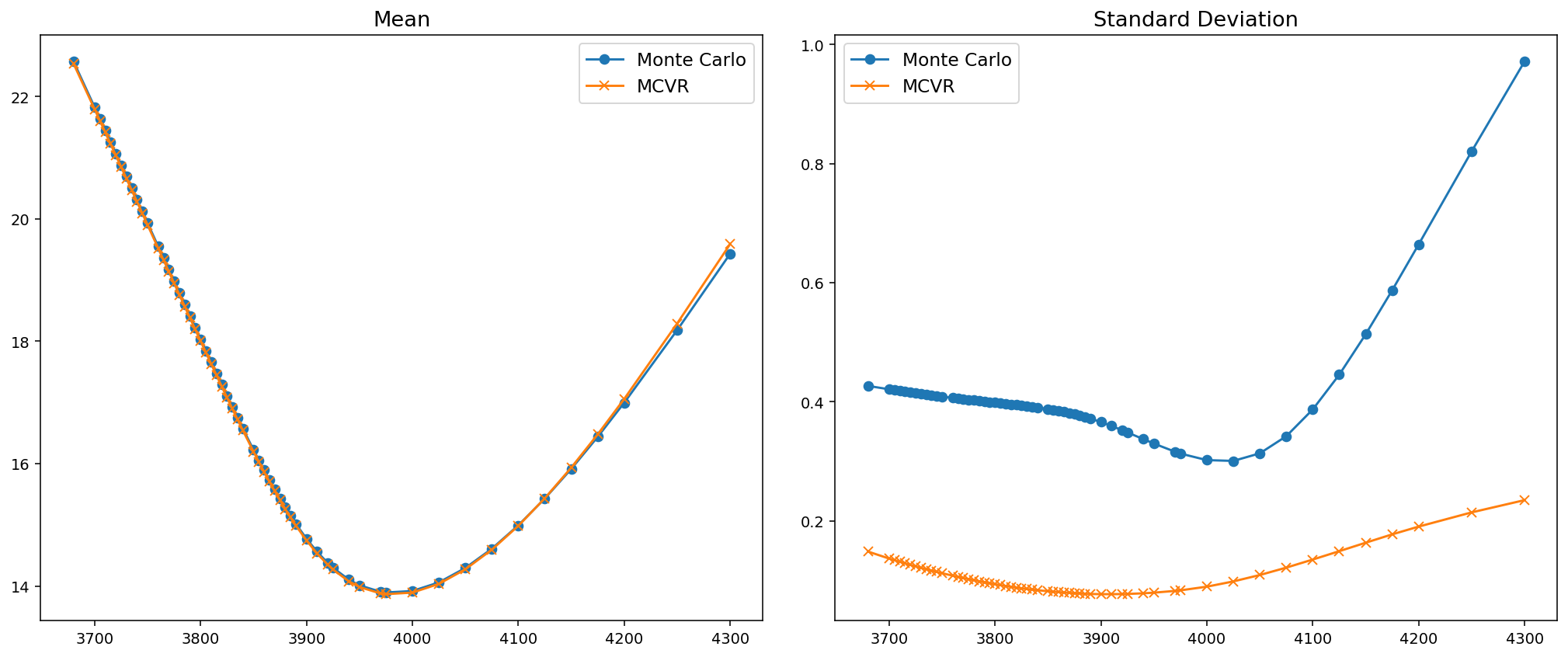}
\caption{Statistics comparing the various numerical methods - SP500 with 3 states}
\label{sp-VarRed-stats-3s}
\end{figure}
\begin{figure}
\centering
\includegraphics[width=345 pt]{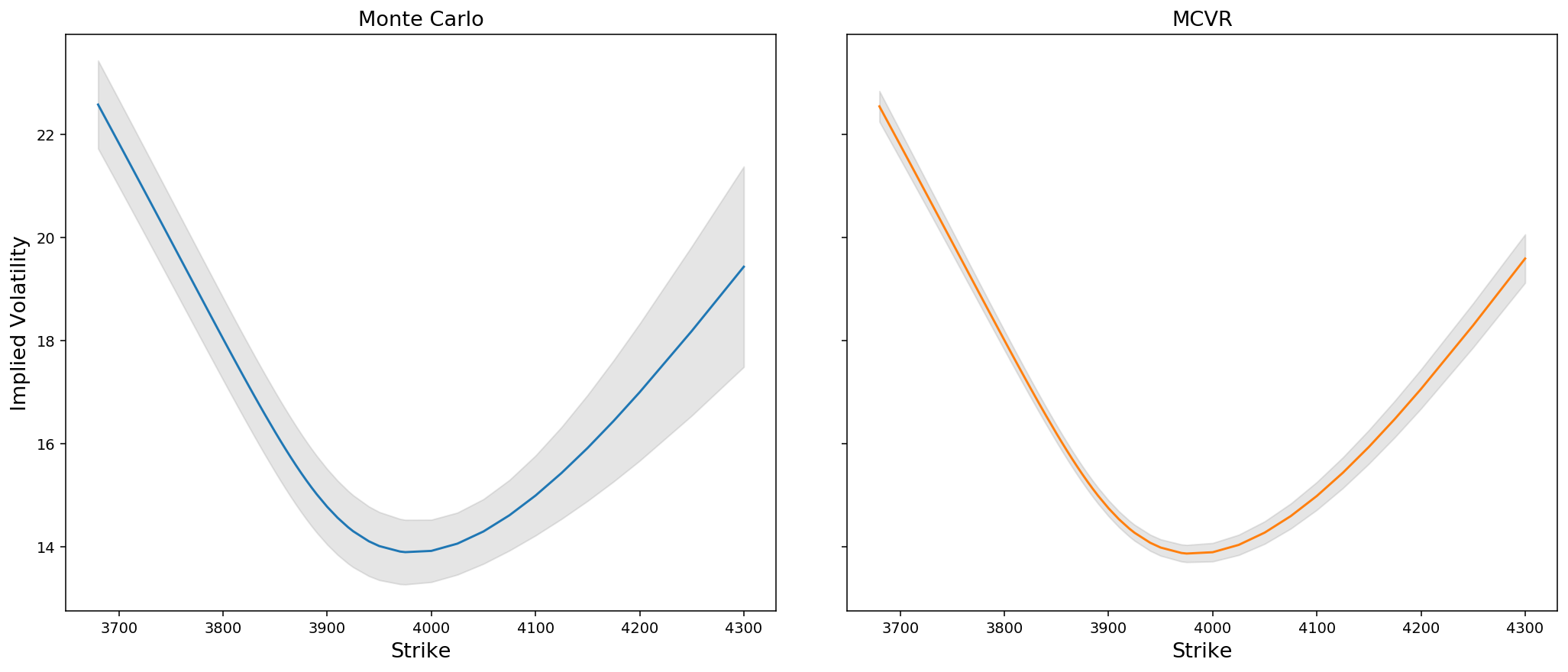}
\caption{Error bars for various numerical methods (2 Standard Deviations) - SP500 with 3 states}
\label{sp-VarRed-errors-3s}
\end{figure}
\FloatBarrier
\subsection{VIX}
In the case of the VIX, we first note in \cref{vix-VarRed-stats} that there is a slight difference between the smiles generated by simple Monte Carlo and the other methods. This is a consequence of the two facts: first, the Control Variate provides a (very good) approximation of the simple Monte Carlo method, but it is not exactly equal; secondly, and most importantly,  the simple Monte Carlo method converges slowly with the number of time steps and we would need a larger number of time steps than the one permitted by the computational budget (which was limited since we had to generate a large number of smiles) for the smiles to be closer. We see that the Control Variate method substantially decreases the variance compared to the simple Monte Carlo, and the MCVR further decreases the variance of the Control Variate method.

\begin{figure}[H]
\centering
\includegraphics[width=345 pt]{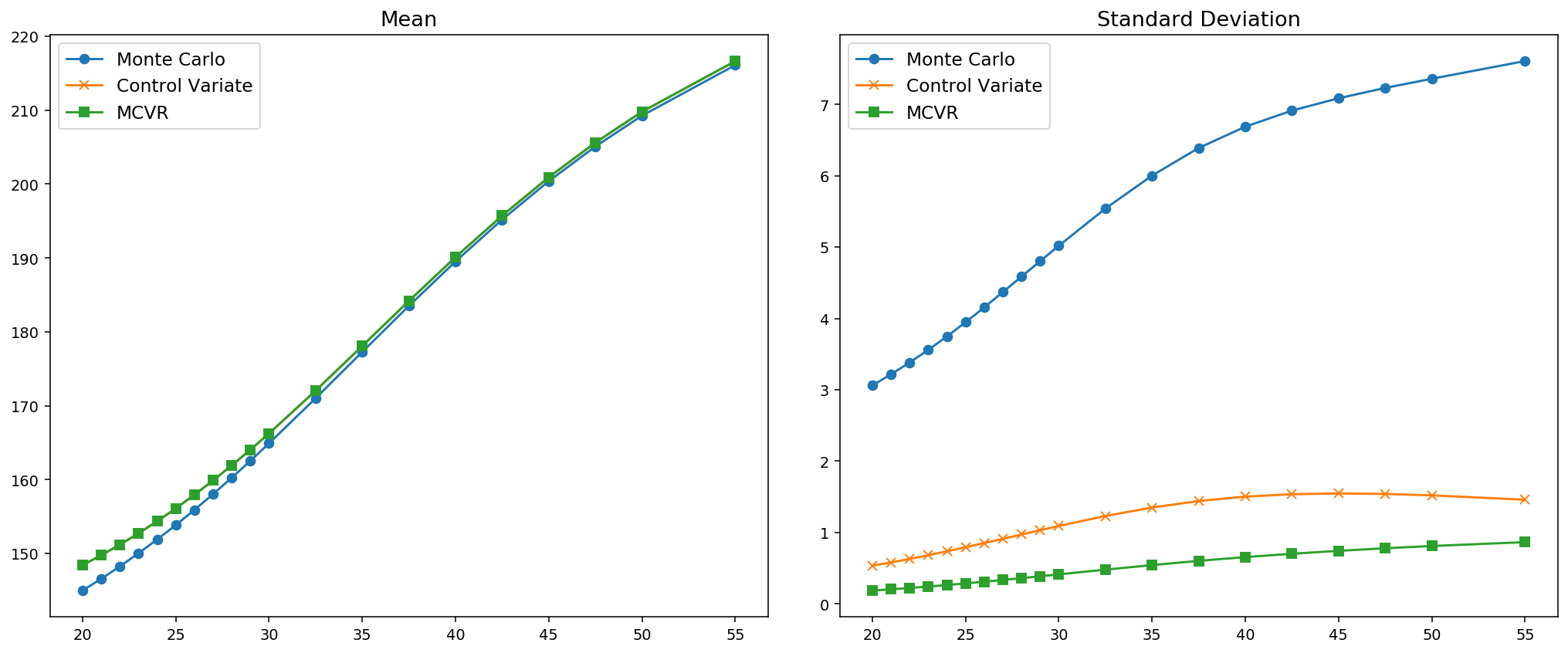}
\caption{Statistics comparing the various numerical methods - VIX}
\label{vix-VarRed-stats}
\end{figure}

\begin{figure}[H]
\centering
\includegraphics[width=345 pt]{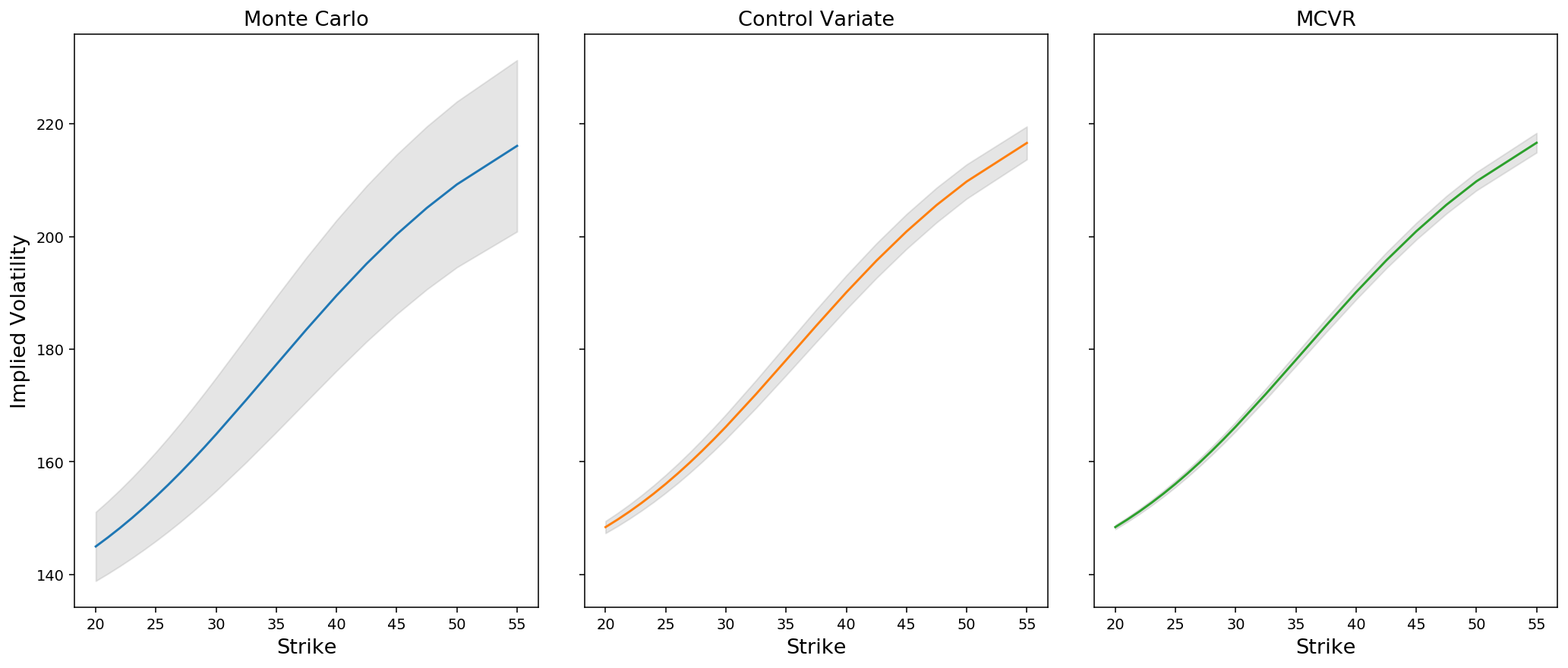}
\caption{Error bars for various numerical methods (2 Standard Deviations - VIX}
\label{vix-VarRed-errors}
\end{figure}

\newpage
\FloatBarrier
\section{Conclusion and further research} \label{sec:conclusion}

In this paper, we proposed a generalization of the rBergomi model, where the change of measure is stochastic and written in terms of a regime-switching \fOU process. From the semi-closed expression for the solution of the fractional SDE, we obtain a semi-closed form for the conditional moment generating function,  which in turn lets us obtain a semi-closed expression for the VIX. 

The fact that the forward variance curve is conditionally log-normal allows for an approximation of the VIX, whose options and future prices are much more efficient to compute. Applying an importance sampling technique to the continuous time Markov chain, allows to further reduce the variance, which leads to fast pricing of VIX options and futures.

The explicit stochastic change of measure we propose allows us to reproduce the upward slopping VIX smiles observed in the market, which the original rBergomi could not do, since it produced flat smiles. The flexibility of the model is also displayed in the fact it can produce good fits to SP500 smiles without falling back into the original rBergomi model. For this reason, the model is even able to provide reasonable fits to the SP500-VIX joint calibration. Thus, the \fOU model we propose has the ability to describe market dynamics both under the physical measure $\mathbb{P}$ and the pricing measure $\mathbb{Q}$.  

We conclude with a few topics for further research. First, to see if it is possible to improve the quality of the fits to observed smiles by increasing the number of possible states (perhaps three). Secondly, this has to be accompanied by either a more efficient calibration procedure or a significant decrease in computing time, since we found the calibration to be too slow and reliant on the initial condition when using three states. Finally, it remains an open question if a different choice for the dynamics of the  long term mean of the \fOU process would be able to provide better fits to observed smiles whilst keeping pricing and calibration computing times at a reasonable level.

\newpage
\appendix 
\section{Appendix - Useful Convolution Results} \label{sec:app-convolution}

In this section, we present some well known convolution results and provide some proofs for the sake of completeness.

\begin{dfn}
Define the translation operator by
\begin{equation}
(\tau_y f)(x) = f(x-y).
\end{equation}
\end{dfn}

\begin{rmk}
It is a known fact that for $1 \leq p < \infty$, the translation operator is continuous on $L^p$ in the sense that for any $f \in L^p$,
\begin{equation}\label{eq:tau-Lp}
\lim_{y \to 0} \norm{ \tau_y f - f}_p = 0.
\end{equation}
\end{rmk}

\begin{lem}\label{lem:cts-conv-general}
Let $1 \leq p \leq \infty$, and $f \in L^p, g \in L^q$, where $1/p + 1/q = 1$. Then $f \star g$ is continuous.
\end{lem}
\begin{proof}
Without loss of generality, assume $p< \infty$, if not, change the roles of $p$ and $q$. To conclude $f \star g$ is continuous, it suffices to show that
\begin{equation}
\lim_{y \to 0} | (\tau_y(f \star g))(x) - (f \star g)(x)| = 0.
\end{equation}
Observe that
\begin{align*}
(\tau_y(f \star g))(x) &= (f \star g)(x-y) \\
&= \int f(x-y-t)g(t) \, dt \\
&= \int (\tau_yf)(x-t)g(t) \, dt \\
&= ((\tau_y f) \star g)(x).
\end{align*}
Thus
\begin{equation}
\tau_y (f \star g) - (f \star g) = (\tau_y f) \star g - f \star g = (\tau_yf - f) \star g.
\end{equation}
Finally, by H{\"o}lder's inequality and the $L^p$-continuity of the translation operator \eqref{eq:tau-Lp}
\begin{equation}
|\tau_y (f \star g) - (f \star g)| \leq \norm{ \tau_y f - f}_p \norm{g}_q \to 0.
\end{equation}
\end{proof}

\begin{lem}\label{lem:cts-conv}
Let $1 \leq p \leq \infty$, and $f \in L^p_{loc}, g \in L^q_{loc}$, where $1/p + 1/q = 1$ and $f, g$ have support on $\R^+$.
Then $f \star g$ is continuous and has support on $\R^+$.
\end{lem}
\begin{proof}
Since $f, g$ have support in $\R^+$, we have
\begin{equation}
\text{supp}(f \star g) = \left\lbrace
x + y \mid x \in \R^+, y \in \R^+
\right\rbrace = \R^+.
\end{equation}
Now we show continuity. Let $T>0$. For a function $h$, let us use the notation $h_T := h \ind{[0,T]}$. It follows from the definition that for any $u \in [0,T]$,
\begin{equation} \label{eq:fTgT}
(f \star g)(u) = \int_0^u f(u-s)g(s) \, ds = 
\int_0^u f_T(u-s)g_T(s) \, ds = 
(f_T \star g_T)(u).
\end{equation}
Because $f \in L^p_{loc}, g \in L^q_{loc}$, it follows that $f_T \in L^p, g_T \in L^q$. By using \cref{lem:cts-conv-general} we know $f_T \star g_T$ is continuous on $\R$. By \eqref{eq:fTgT} this implies that $f \star g$ is continuous on $[0,T)$. Since $T$ was arbitrary, we conclude that $f \star g$ is continuous on $\R_0^+$. Finally, since $f \star g$ has support on $\R^+$ and $(f \star g)(0)=0$, it is in fact continuous in $\R$.
 
\end{proof}

\begin{dfn} \label{dfn:resolvent}
For a kernel $K \in L^1_{loc}(\R^+, \R^{d \times d})$, the resolvent of $K$ (also called resolvent of the second kind)  is the unique kernel $R \in L^1_{loc}(\R^+, \R^{d \times d})$ that solves the resolvent equation
\begin{equation}\label{eq:resolvent-eq}
K \star R = R \star K = K-R.
\end{equation}
For more details, see \cite{affine}.
\end{dfn}

\begin{rmk} \label{rmk:Gripen}
The resolvent inherits various properties of the original kernel $K$. Indeed, if $K \in L^p_{loc}$, with $1  \leq p \leq +\infty$, then also $R \in L^p_{loc}$. For more details, see \cite[Theorem 3.5]{Gripen1990}.
\end{rmk}

\begin{lem} \label{lem:eB-exact}
Let $K$ be the fractional kernel $K(x) = x^{\alpha-1}$, where $\alpha \in (1/2,1)$. Let $\theta \in \R \setminus \{0\}$. Let $R_\theta$ be the resolvent of $K\theta$, and $E_\theta = K - R_\theta \star K$. Then
\begin{equation}\label{eq:RB-exact}
R_\theta(t) = \theta t^{\alpha-1}\psi(t) 
\end{equation}
and
\begin{equation}\label{eq:eB-exact}
E_\theta(t) =t^{\alpha-1} \psi(t),
\end{equation}
where $\psi$ is the continuous function defined on $\R_0^+$ by
\begin{equation}
\psi(t) = \Gamma(\alpha) E_{\alpha, \alpha}(-\theta \Gamma(\alpha)t^\alpha)
\end{equation}
and $E$ denotes the \ML{}
\begin{equation} \label{eq:ml-function}
E_{\alpha, \beta}(z) = \sum_{n=0}^{+\infty} \frac{z^n}{\Gamma(\alpha n + \beta)}.
\end{equation}

\end{lem}
\begin{proof}
It is a known fact (see \cite[Table~1]{affine}) that the resolvent of $ct^{\alpha-1}/\Gamma(\alpha)$ is $ct^{\alpha-1}E_{\alpha, \alpha}(-ct^\alpha)$, for $c \in \R$. Thus, the resolvent of $K\theta$ is
\begin{equation} 
R_\theta(t) = \theta \Gamma(\alpha)t^{\alpha-1}E_{\alpha, \alpha}(-\theta \Gamma(\alpha)t^\alpha) = -t^{\alpha-1}B\psi(t).
\end{equation}
By the resolvent equation \eqref{eq:resolvent-eq} and noting we are working in a one-dimensional setting it follows that
\begin{align}
\theta E_\theta &= \theta K + B(R_\theta \star K) \\
&= K\theta - R_\theta \star (K\theta) \\
&= K\theta - (K\theta  - R_\theta) \\
&= R_\theta.
\end{align}
Hence
\begin{equation}
E_\theta = \frac{1}{\theta} R_\theta = t^{\alpha-1}\psi(t).
\end{equation}
The fact that $\psi$ is continuous follows trivially from the fact that $\alpha >0$ and the \ML{} is an entire function as long as $\alpha, \beta \in \R^+$.

\end{proof}

\begin{lem}
Let $\alpha \in (1/2,1]$. Then
\begin{equation}
\int_0^u  t^{\alpha-1} E_{\alpha, \alpha}(-t^\alpha) \, dt
= 1 - E_{\alpha, 1}(-u^\alpha).
\end{equation}
\end{lem}

\begin{proof}
By definition, 
\begin{equation}
E_{\alpha, \alpha}(-t^\alpha) = \sum_{n=0}^{+\infty} (-1)^n\frac{t^{\alpha n}}{\Gamma(\alpha(n+1))}.
\end{equation}
Hence
\begin{equation}
t^{\alpha-1} E_{\alpha, \alpha}(-t^\alpha) =  \sum_{n=0}^{+\infty} (-1)^n\frac{t^{\alpha (n+1)-1}}{\Gamma(\alpha(n+1))}.
\end{equation}

Provided we can apply Fubini's theorem:
\begin{align*}
\int_0^u  t^{\alpha-1} E_{\alpha, \alpha}(-t^\alpha) \, dt &= 
\sum_{n=0}^{+\infty} 
\int_0^u 
\frac{(-1)^n}{\Gamma(\alpha(n+1))}
t^{\alpha (n+1)-1} 
\, dt \\
&= \sum_{n=0}^{+\infty} 
(-1)^n
\frac{u^{\alpha(n+1)}}{\alpha(n+1)\Gamma(\alpha(n+1))} \\
&=  - \sum_{k=1}^{+\infty} (-1)^k
\frac{u^{\alpha k}}{\alpha k\Gamma(\alpha k)}  \\
&=  1- \sum_{k=0}^{+\infty} (-1)^k
\frac{u^{\alpha k}}{\Gamma(\alpha k +1)} \\
&= 1 - E_{\alpha, 1}(-u^\alpha).
\end{align*}
The application of Fubini's theorem is justified by the fact that, proceeding as above, we obtain
\begin{equation}
\sum_{n=0}^{+\infty} 
\int_0^u 
\frac{1}{\Gamma(\alpha(n+1))}
t^{\alpha (n+1)-1} 
\, dt = E_{\alpha, 1}(u^\alpha) -1 < \infty.
\end{equation}

\end{proof}

\begin{lem} \label{lem:eB-integral}
Let $R_\theta$ be as in \cref{lem:eB-exact}, with $\theta>0$. Then
\begin{equation}
\int_0^u R_\theta(t) \, dt = 1- E_{\alpha, 1}\left(-\theta\Gamma(\alpha) u^\alpha\right).
\end{equation}
\end{lem}
Let $c=\theta\Gamma(\alpha)>0$. Make the change of variables $x=c^{1/\alpha}t$, so that $dt = c^{1-1/\alpha} dx$. Then
\begin{align*}
\int_0^u R_\theta(t) \, dt &= \int_0^{c^{1/\alpha}u} x^{\alpha-1} 
E_{\alpha, \alpha}(- x^\alpha) \, dx \\
&=  \int_0^{uc^{1/\alpha}} x^\alpha 
E_{\alpha, \alpha}(- x^\alpha) \, dx  \\
&= 1 - E_{\alpha,1}(-c u^\alpha).
\end{align*}

\newpage
\section{Girsanov change of measure} \label{sec:girsanov}

The Brownian motions involved in our model are written in terms of the three-dimensional Brownian motion
\begin{equation}
\textbf{B}^\mathbb{P} = (\bar{W}^\mathbb{P}, Z^\mathbb{P},  \bar{Z}^\mathbb{P}).
\end{equation}
For this Brownian motion, a general Girsanov change of measure is given by
\begin{equation}
\frac{d\mathbb{Q}}{d\mathbb{P}} = M_T,
\end{equation}
where
\begin{equation}
M_t = \mathcal{E}\left( \int \alpha \cdot d\textbf{B}^\mathbb{P} \right)_t = \exp
\left(
\int_0^t \alpha_s \cdot d\textbf{B}_s^\mathbb{P} - 
\frac{1}{2}
\int_0^t
\norm{\alpha_s}^2 ds
\right),
\end{equation}
and
\begin{equation}
\alpha_s = (a_s, b_s, c_s),
\end{equation}
for adapted processes $a, b, c$.

Intuitively, we only need to apply the change of measure to two sources of randomness: the price and volatility drivers. Thus, the Brownian motion $Z$ may be left unchanged and we may set $b_t \equiv 0$. The above change of measure implies that
\begin{equation}
a_t = \frac{1}{\bar{\rho}}(u_t + \rho \lambda_t),
\end{equation}
and
\begin{equation}
c_t = - \frac{1}{\bar{\eta}} \lambda_t,
\end{equation}
where $\vartheta$ is the market price of risk $\vartheta_t = \zeta_t/\sqrt{v_t}>0$ and $\lambda_t$ is the change of measure for the volatility component, as in \eqref{eq:lambda-dW}. In the case of our \fOU regime switching change of measure, the process $\lambda_t$ is given by \eqref{eq:lambda-fOU}. To ensure $\mathbb{P} \sim \mathbb{Q}$, we still have to verify that
\begin{equation}
\evM{P}{ M_T } = 1.
\end{equation}
To this end, we adapt the proof from the ones found in \cite[Lemma 7.3]{affine} and \cite[Appendix C]{Jaber2020MarkowitzPS}. Consider the stopping times
\begin{equation}
\tau_n = \inf \{
t > 0 \mid \max_{i=1,2} U_t^{(i) } > n
\} \land T,
\end{equation}
where 
\begin{equation}
U_t^1 = \int_0^t \vartheta_s^2 \, ds 
\end{equation}
and \begin{equation}
U_t^2 = \int_0^t X_s^2 \, ds. 
\end{equation}
Consider now the processes $\alpha_n$ defined by
\begin{equation}
\alpha^n_s = \ind{ \{ s \leq \tau_n \}} \alpha_s
\end{equation}
and the corresponding sequence of measures $\mathbb{Q}^n$  
\begin{equation}
\frac{d\mathbb{Q}^n}{d \mathbb{P}} = M^{\tau_n}_T = M_{\tau_n},
\end{equation}
where the process $M^{\tau_n}$ is defined by
\begin{equation}
M^{\tau_n}_t = \mathcal{E}\left( \int \alpha^n \cdot d\textbf{B}^\mathbb{P} \right)_t.
\end{equation}
We have not specified the dynamics of the market price of risk $\vartheta$, since it is outside the scope of this paper, but we will assume it can be controlled in the same way as $X$ in the following sense:
\begin{assumption} \label{ass:market-price}
There exists a constant $C_2^\prime$, which does not depend on $n$, such that for all $n \geq 1$,
\begin{equation}
\sup_{0 \leq t \leq T} \mathbb{Q}^n \left[
\vartheta_t^2
 \right] \leq C_2^\prime.
\end{equation}
\end{assumption}
For more details concerning this assumption and some sufficient conditions to satisfy it, see \cref{rmk:ass-market-price}.

We can now verify Novikov's condition for $\alpha^n$. Indeed
\begin{align*}
\norm{ \alpha_t^n}^2 &=  
 a_t^2 + c_t^2 \\
 &= \frac{1}{1-\rho^2}(\vartheta_t + \rho \lambda_t)^2 + \frac{1}{1-\eta^2} \lambda_t^2  \\
&\leq  \frac{2}{1-\rho^2} (\vartheta_t^2 + \rho^2 \lambda_t^2) + \frac{1}{1-\eta^2} \lambda_t^2,
\end{align*}
where we used the useful inequality $(a+b)^2 \leq 2(a^2 + b^2)$ for $a, b \in \R$. Recall that $\lambda_t = \theta(\mu_t - X_t)$, where $\mu$ follows a CTMC and hence is bounded by a certain number $C_\mu \in \R$. Thus
\begin{equation}
\lambda_s^2 \leq 2\theta^2( C_\mu^2 + X_t^2).
\end{equation}
Putting it all together, it follows that
\begin{equation}
\norm{ \alpha_t^n}^2 \leq \ind{ \{ s \leq \tau_n \}} \left( \beta_0 + \beta_1 \vartheta_t^2 + \beta_2 X_t^2 \right),
\end{equation}
where the constants $\beta_0, \beta_1, \beta_2$ are given by
\begin{equation}
\begin{split}
\beta_0 &= 2\theta^2C_\mu^2 \left( 
\frac{2 \rho^2}{1-\rho^2} + \frac{1}{1-\eta^2}
\right), \\
\beta_1 &= \frac{2}{1-\rho^2}, \\
\beta_2 &= \frac{1}{C_\mu^2}\beta_0.
\end{split}
\end{equation}
Hence
\begin{equation}
\int_0^T \norm{ \alpha^n_s }^2 \, ds \leq \beta_0T + \beta_1  U^1_{\tau_n}+ \beta_2 U^2_{\tau_n}.
\end{equation}
By construction, $U^i_{\tau_n} \leq n$ a.s. for $i=1,2$. 
Thus, Novikov's condition is verified and therefore the process $M^{\tau_n}$ is a true martingale. Hence
\begin{equation} \label{aux:novikov-n}
1 = M^{\tau_n}_0 = \ev{ M^{\tau_n}_T} = \ev{M_{\tau_n}} = \ev{
M_{\tau_n} \ind{ \{ \tau_n < T \} }
} + \ev{
M_T \ind{ \{ \tau_n = T \}}
}
\end{equation}
Since $M$ is a supermartingale, $M_T \in L^1$. Moreover, since $X$ is continuous, $U_T < \infty$ a.s. and thus $\tau_n \to T$ a.s. when $n \to \infty$. By the dominated convergence theorem it follows that
\begin{equation}
\ev{ M_T \ind{ \{ \tau_n = T \}} } \to \ev{M_T}.
\end{equation}
We are left to show that
\begin{equation}
\ev{M_{\tau_n} \ind{ \{ \tau_n < T \} }} \to 0.
\end{equation}
By construction, we have that
\begin{align*}
\ev{M_{\tau_n} \ind{ \{ \tau_n < T \} }} &=
 \mathbb{Q}^n\left[ \tau_n < T \right]  \\
&\leq \sum_{i=1}^2 \mathbb{Q}^n\left[ U_T^i > n \right] \\
&\leq \frac{1}{n}  \sum_{i=1}^2 \mathbb{Q}^n \left[
|U_T^{(i)}|
\right],
\end{align*}
where we used the Markov inequality in the last step.

For more general processes, we could proceed here as in \cite[Lemma 7.3]{affine}, but since we know our process $X$ exactly we will simplify the argument. Indeed, by \eqref{eq:sol-time-dep}, we know that
\begin{equation}
|X_t| \leq f(t) + \sigma |Y_t|.
\end{equation}
where $f$ is the continuous function
\begin{equation}
f(t) = |g(t)| + C_\mu \int_0^t E_\theta(s) \, ds.
\end{equation}
Since $Y$ is Gaussian and its variance is a continuous function, described in \eqref{eq:et}, it follows that there exists a constant $C_2 \in \R$ such that
\begin{equation}
\sup_{0 \leq t \leq T} \ev{|X_t|^2} \leq C_2.
\end{equation}
By the Fubini-Tonelli theorem
\begin{equation}
\ev{ |U_T^{(2)}|} = \ev{ \int_0^T X_s^2 \, ds } \leq T \sup_{0 \leq s \leq T} \ev{X_s^2} \leq TC_2.
\end{equation}
The sBm $Z$ is also a $\mathbb{Q}^n$-sBm since $b_t \equiv 0$ in the change of measure. Thus, the above constant does not depend on $n$. Finally, note that \cref{ass:market-price} guarantees that
\begin{equation}
\ev{ |U_T^{(1)}|}  \leq T C_2'.
\end{equation}
\begin{rmk}
The fact that the Brownian motion $Z$ is not affected by the change of measure is by no means essential to the argument, as it can be seen in the proof of \cite[Lemma 7.3]{affine} and also in the argument for $\vartheta$ in \cref{rmk:ass-market-price} bellow.
\end{rmk}

\begin{rmk} \label{rmk:ass-market-price}
The \cref{ass:market-price} may look \textit{ad hoc}, but note that it will be satisfied if $\vartheta$ itself satisfies an equation similar to that of $X$. Indeed, assume 
\begin{equation}
\vartheta = \vartheta_0 + K \star (q \, dt + \sigma d\textbf{B}^\mathbb{P}).
\end{equation}
The above is equivalent to
\begin{equation}
\vartheta = \vartheta_0 + K \star ( \tilde{q} \, dt + \sigma d\textbf{B}^{\mathbb{Q}^n}),
\end{equation}
where \begin{equation}
\tilde{q}(t, x, \omega) = q(t, x, \omega) + \sigma \ind{ \{t \leq \tau_n\} }(\omega).
\end{equation}
Provided $q$ satisfies a uniform linear growth condition on $x$, so will $\tilde{q}$. Moreover, since $\alpha^n$ verifies Novikov's condition, we are guaranteed that $\textbf{B}^{\mathbb{Q}^n}$ is a $\mathbb{Q}_ n$-sBm. The existence of the constant $C_2'$ will then be guaranteed again by \cite[Lemma 3.1, Remark 3.2]{affine}.

Note also that the assumption will also be satisfied if $\vartheta$ can be written as a function of such process (which can even be $X$) and that function itself satisfies the linear growth condition.
\end{rmk}

\newpage
\section{Weakly singular kernels} \label{sec:num-integral}
When dealing with fractional processes, it is common to encounter integrals with weakly singular kernels of the form
\begin{equation}
\int_0^u s^x (u-s)^y \varphi(s) \, ds,
\end{equation}
where $-1 < x,y < 0$ and $\varphi$ is a continuous function on $[0,u]$. To approximate such integrals, we use an idea similar to \cite{Hybrid2017} and write them as
\begin{equation}
\int_0^u s^x (u-s)^y \varphi(s) \, ds = I_1 + I_2 + I_3,
\end{equation}
where
\begin{equation}
\begin{split}
I_1 &= \int_0^\varepsilon s^x (u-s)^y \varphi(s) \, ds, \\
I_2 &= \int_\varepsilon^{u-\varepsilon} s^x (u-s)^y \varphi(s) \, ds, \\
I_3 &= \int_{u-\varepsilon}^u s^x (u-s)^y \varphi(s) \, ds,
\end{split}
\end{equation}
for a small $\varepsilon>0$, which is usually taken to be the grid time step. The integral $I_2$ can be approximated by a standard quadrature method since the integrand does not contain any singularities. For integrals $I_1$ and $I_3$, we approximate the non-singular part by a constant and integrate the singular part analytically:
\begin{equation}
\begin{split}
I_1 &\approx \frac{\phi(0)u^y + \phi(\varepsilon)(u-\varepsilon)^y }{2} \frac{\varepsilon^{x+1}}{x+1}, \\
\\
I_3 &\approx \frac{\phi(u-\varepsilon)(u-\varepsilon)^x + \phi(u)u^x }{2} \frac{\varepsilon^{y+1}}{y+1}.
\end{split}
\end{equation}

\newpage
\bibliography{bib}

\begin{thebibliography}{}

\bibitem[Abi~Jaber, 2020]{Jaber2020TheCF}
Abi~Jaber, E. (2020).
\newblock The characteristic function of {G}aussian stochastic volatility
  models: an analytic expression.
\newblock {\em arXiv preprint arXiv:2009.10972}.

\bibitem[Abi~Jaber, 2021]{Jaber2019WeakEA}
Abi~Jaber, E. (2021).
\newblock {Weak existence and uniqueness for affine stochastic {V}olterra
  equations with ${L^{1}}$-kernels}.
\newblock {\em Bernoulli}, 27(3):1583 -- 1615.

\bibitem[Abi~Jaber et~al., 2019a]{Jaber2019AWS}
Abi~Jaber, E., Cuchiero, C., Larsson, M., and Pulido, S. (2019a).
\newblock A weak solution theory for stochastic {V}olterra equations of
  convolution type.
\newblock {\em arXiv preprint arXiv:1909.01166}.

\bibitem[Abi~Jaber et~al., 2019b]{affine}
Abi~Jaber, E., Larsson, M., and Pulido, S. (2019b).
\newblock {Affine Volterra processes}.
\newblock {\em The Annals of Applied Probability}, 29(5):3155 -- 3200.

\bibitem[Abi~Jaber et~al., 2021]{Jaber2020MarkowitzPS}
Abi~Jaber, E., Miller, E., and Pham, H. (2021).
\newblock Markowitz portfolio selection for multivariate affine and quadratic
  {V}olterra models.
\newblock {\em SIAM Journal on Financial Mathematics}, 12(1):369--409.

\bibitem[Ackermann et~al., 2020]{InhomVolterra}
Ackermann, J., Kruse, T., and Overbeck, L. (2020).
\newblock Inhomogeneous affine {V}olterra processes.
\newblock {\em arXiv preprint arXiv:2012.10966}.

\bibitem[Alos et~al., 2018]{Alos18}
Alos, E., Garc{\'\i}a-Lorite, D., and Muguruza, A. (2018).
\newblock On smile properties of volatility derivatives and exotic products:
  understanding the {VIX} skew.
\newblock {\em arXiv preprint arXiv:1808.03610}.

\bibitem[Al{\`o}s et~al., 2007]{Alos2007}
Al{\`o}s, E., Le{\'o}n, J.~A., and Vives, J. (2007).
\newblock On the short-time behavior of the implied volatility for
  jump-diffusion models with stochastic volatility.
\newblock {\em Finance and Stochastics}, 11(4):571--589.

\bibitem[Bayer et~al., 2016]{PricingRough}
Bayer, C., Friz, P., and Gatheral, J. (2016).
\newblock Pricing under rough volatility.
\newblock {\em Quantitative Finance}, 16(6):887--904.

\bibitem[Bennedsen et~al., 2017]{Hybrid2017}
Bennedsen, M., Lunde, A., and Pakkanen, M.~S. (2017).
\newblock Hybrid scheme for {B}rownian semistationary processes.
\newblock {\em Finance and Stochastics}, 21(4):931--965.

\bibitem[{Chicago Board Options Exchange}, 2019]{VIX-WhitePaper}
{Chicago Board Options Exchange} (2019).
\newblock {VIX}: {CBOE} volatility index.

\bibitem[Comte et~al., 2012]{Comte2012AffineFS}
Comte, F., Coutin, L., and Renault, {\'E}. (2012).
\newblock Affine fractional stochastic volatility models.
\newblock {\em Annals of Finance}, 8:337--378.

\bibitem[El~Euch et~al., 2018]{Microstructure}
El~Euch, O., Fukasawa, M., and Rosenbaum, M. (2018).
\newblock The microstructural foundations of leverage effect and rough
  volatility.
\newblock {\em Finance and Stochastics}, 22(2):241--280.

\bibitem[El~Euch et~al., 2019]{rHeston}
El~Euch, O., Gatheral, J., and Rosenbaum, M. (2019).
\newblock Roughening heston.
\newblock {\em Risk Management \& Analysis in Financial Institutions eJournal},
  pages 84--89.
\newblock Available at SSRN: \url{https://ssrn.com/abstract=3116887} or
  \url{http://dx.doi.org/10.2139/ssrn.3116887}.

\bibitem[El~Euch and Rosenbaum, 2019]{rHeston-Char}
El~Euch, O. and Rosenbaum, M. (2019).
\newblock The characteristic function of rough {H}eston models.
\newblock {\em Mathematical Finance}, 29(1):3--38.

\bibitem[Filipovi{\'c}, 2005]{Filipovi2005TimeinhomogeneousAP}
Filipovi{\'c}, D. (2005).
\newblock Time-inhomogeneous affine processes.
\newblock {\em Stochastic Processes and their Applications}, 115:639--659.

\bibitem[Fukasawa, 2020]{Fukasawa2020VolatilityHT}
Fukasawa, M. (2020).
\newblock Volatility has to be rough.
\newblock {\em Quantitative Finance}, 21:1 -- 8.

\bibitem[Gatheral et~al., 2018]{VolIsRough}
Gatheral, J., Jaisson, T., and Rosenbaum, M. (2018).
\newblock Volatility is rough.
\newblock {\em Quantitative Finance}, 18(6):933--949.

\bibitem[Gatheral et~al., 2020]{Joint-VIX-rHeston}
Gatheral, J., Jusselin, P., and Rosenbaum, M. (2020).
\newblock The quadratic rough {H}eston model and the joint {S}\&{P}500{/VIX}
  smile calibration problem.
\newblock {\em arXiv preprint arXiv:2001.01789}.

\bibitem[Gripenberg et~al., 1990]{Gripen1990}
Gripenberg, G., Londen, S.~O., and Staffans, O. (1990).
\newblock {\em Volterra Integral and Functional Equations}.
\newblock Encyclopedia of Mathematics and its Applications. Cambridge
  University Press.

\bibitem[Guerreiro and Guerra, 2021]{LSMC}
Guerreiro, H. and Guerra, J. (2021).
\newblock Least squares {M}onte {C}arlo methods in stochastic {V}olterra rough
  volatility models.
\newblock {\em arXiv preprint arXiv:2105.04511}.

\bibitem[Hinsen, 2017]{ML-Python}
Hinsen, K. (2017).
\newblock {The Mittag-Leffler function in Python}.
\newblock \url{https://github.com/khinsen/mittag-leffler}.

\bibitem[Horvath et~al., 2020]{ModulatedVolterra}
Horvath, B., Jacquier, A., and Tankov, P. (2020).
\newblock Volatility options in rough volatility models.
\newblock {\em SIAM Journal on Financial Mathematics}, 11:437--469.

\bibitem[Livieri et~al., 2018]{Livieri2018RoughVE}
Livieri, G., Mouti, S., Pallavicini, A., and Rosenbaum, M. (2018).
\newblock Rough volatility: Evidence from option prices.
\newblock {\em IISE Transactions}, 50:767 -- 776.

\bibitem[Newville et~al., 2014]{LMFIT}
Newville, M., Stensitzki, T., Allen, D.~B., and Ingargiola, A. (2014).
\newblock {LMFIT: Non-Linear Least-Square Minimization and Curve-Fitting for
  Python}.
\newblock DOI: \url{https://doi.org/10.5281/zenodo.598352}.

\bibitem[Wang et~al., 2021]{Wang-fOU}
Wang, X., Xiao, W., and Yu, J. (2021).
\newblock Modeling and forecasting realized volatility with the fractional
  {Ornstein–Uhlenbeck} process.
\newblock {\em Journal of Econometrics}.
\newblock DOI: \url{https://doi.org/10.1016/j.jeconom.2021.08.001}.

\bibitem[Wang, 2008]{Wang2008ExistenceAU}
Wang, Z. (2008).
\newblock Existence and uniqueness of solutions to stochastic {V}olterra
  equations with singular kernels and non-{L}ipschitz coefficients.
\newblock {\em Statistics \& Probability Letters}, 78:1062--1071.

\end{thebibliography}

\end{document}